\documentclass[onecolumn,journal]{IEEEtran}
\usepackage{amsmath,amsxtra,amssymb,amsthm,amsfonts,eufrak,bm,dsfont}
\usepackage{hyperref}
\usepackage{cleveref}
\usepackage{diagbox}
\usepackage[noadjust]{cite}
\usepackage{tikz-cd}
\usepackage{float}
\usepackage{mathrsfs}
\usetikzlibrary{matrix,arrows,decorations.pathmorphing}
\usepackage{soul}

\newtheorem{lemma}{Lemma}
\newtheorem{prop}[lemma]{Proposition}
\newtheorem{theorem}[lemma]{Theorem}

\newtheorem{rem}[lemma]{Remark}

\newcommand{\re}{\begin{rem}\rm}
\newcommand{\mar}{\end{rem}}
\newtheorem{exam}[lemma]{Example}

\newcommand{\qd}{\end{proof}\vspace{0.5ex}}
\newcommand{\prf}{\begin{proof}[\bf Proof:]}
\newcommand{\frS}{\mathfrak{S}}
\def\blambda{\bm{\lambda}}

\def\SU{\mathop{\rm SU}}

\newcommand{\pl}{\hspace{.1cm}}

\newcommand{\al}{\alpha}

\newcommand{\eps}{\varepsilon}

\newcommand{\norm}[2]{\parallel \! #1 \! \parallel_{#2}}

\newcommand{\ketbra}[1]{|{#1}\rangle\langle{#1}|}

\newcommand{\ten}{\otimes}

\DeclareMathOperator{\Tr}{\operatorname{Tr}}

\def\argmax{\mathop{\rm argmax}}

\newcommand{\cH}{\mathcal{H}}

\allowdisplaybreaks

\def\Label#1{\label{#1}\ [\ \text{#1}\ ]\ }
\def\Label{\label}

\begin{document}

\title{Resolvability of classical-quantum channels}

\author{Masahito Hayashi, % \IEEEmembership{Fellow, IEEE},
Hao-Chung Cheng,
and Li Gao
\thanks{The work of MH was supported
in part by
the National Natural Science Foundation of
China (Grant No.~62171212).
HC is supported by NSTC 113-2119-M-001-009, No.~NSTC 113-2628-E-002-029, No.~NTU-113V1904-5, No.~NTU-CC-113L891605, and No.~NTU-113L900702. LG is partially supported by the National Natural Science Foundation of
China.
}
\thanks{Masahito Hayashi is with
School of Data Science, The Chinese University of Hong Kong,
Shenzhen, Longgang District, Shenzhen, 518172, China
%International Quantum Academy, Futian District, Shenzhen 518048, China,
and
the Graduate School of Mathematics, Nagoya University, Nagoya, 464-8602, Japan
(E-mail: \texttt{\href{mailto:hmasahito@cuhk.edu.cn}{hmasahito@cuhk.edu.cn}}
).
Li Gao is with
the  School of Mathematics and Statistics, Wuhan University, Wuhan, 430072, China
(E-mail: \texttt{\href{mailto:gao.li@whu.edu.cn}{gao.li@whu.edu.cn}}
).
Hao-Chung Cheng is with the Department of Electrical Engineering, the Graduate Institute of Communication Engineering, the Department of Mathematics, and the Institute of Applied Mathematical Sciences, and the Center for Quantum Science and Engineering, National Taiwan University, Taipei 10617, Taiwan,
and the Physics Division, National Center for Theoretical Sciences, Taipei 10617, Taiwan
and with the Hon Hai (Foxconn) Quantum Computing Center, New Taipei City 236, Taiwan
(E-mail: \texttt{\href{mailto:haochung@ntu.edu.tw}{haochung@ntu.edu.tw}}).}

}
\markboth{M.~Hayashi,
L.~Gao, and H.-C.~Cheng: Resolvability of CQ channels}{}

\maketitle

\begin{abstract}
Channel resolvability concerns the minimum resolution for approximating the channel output.
We study the resolvability of classical-quantum channels in two settings, for the channel output generated from the worst input, and form the fixed independent and identically distributed (i.i.d.) input.
The direct part of the worst-input setting is derived from sequential hypothesis testing as it involves of non-i.i.d.~inputs. The strong converse of the worst-input setting is obtained via the connection to identification codes.
For the fixed-input setting, while the direct part follows from
the known quantum soft covering result, % \cite{hayashi2012quantum, cheng2023error},
 we exploit the recent alternative quantum Sanov theorem to solve the strong converse. %st \cite{another}.
\end{abstract}

\begin{IEEEkeywords}
Channel resolvability,
classical-quantum channel,
identification code,
worst-input formulation,
fixed-input formulation
\end{IEEEkeywords}

\section{Introduction} \label{sec:introduction}

%% Classical channel resolvability and its applications (literature)
Channel resolvability is a problem to approximate or simulate the channel output by using limited number of input symbols.
For classical channels, this problem dates back to Wyner \cite{Wyn75b} for investigating common information between correlated random variables \cite{Wyn75b, Cuff08, YT18, YT20, YT22_book}, and Han and Verd{\'u}'s study of approximating output statistics \cite{han1993approximation}, \cite[§ 3]{Han03}.
The nature of this problem is closed related the fundamental concept of \emph{covering}.
Hence it constitutes a key tool in information theory with applications to a wealth of cryptographic and simulation problems, including
secrecy for a wiretap channel \cite{Wyn75, CK78, hayashi2006general2, hayashi2011, BL13, PTM17},
identification via channels \cite{AD89},
channel synthesis based on the soft covering lemma \cite{Cuf13, Cuff, LCV16, PTM17, YaCu, Yassaee},
rate distortion theory \cite{SV96},
random binning \cite{YAM+14, MBG+19},
as well as noise stability \cite{YT22_book, Yu24}.
The aforementioned covering problems can be broadly subsumed under the umbrella of the reverse Shannon theorem \cite{BSS+02, BCR11, BDH+14, OCC+24} as the problem nature is to simulate a stochastic object.
This is dual to Shannon's noisy coding theorem for communication over noisy channels as a packing problem, for which the aim is to discriminate among stochastic objects.

%% Classical-Quantum Chanels (literature)
The history of the channel resolvability for classical-quantum (CQ) channel goes back to the papers \cite{1377491,CWY}, which  discussed approximations of the output states under
CQ wiretap channels.
The work \cite{hayashi2012quantum} explicitly formulated the problem of resolvability of CQ channels with a fixed input distribution,
and obtained the exponential decay of the approximation error in (independent and identically distributed) i.i.d.~limits.
Later, the paper \cite{hayashi2012quantum} applied it to CQ wiretap channels for the derivation of the exponential decreasing rate of information leakage.
Also, the papers \cite{loeber,ahlswede2002strong} addressed identification codes for CQ channels based on a probability deviation bound for operators (see also \cite[§ 17]{Wilde17}).
Recently, several researchers studied the topic of soft covering for CQ channels \cite{DW05, Win04, Win05, CWY, 1377491, hayashi2012quantum,  cheng2023error, SGC22b, SGC23}, \cite[§ 9]{hbook},
which gives a random coding approach to direct part of channel resolvability \cite{hayashi2012quantum}.
Such a technique applies to secret key generation from quantum states \cite{DW05, Win05, CHW16, KKG+19},
measurement compression \cite{Win04},
private communication over quantum channels \cite{1377491, CWY, WNI, hayashi2012quantum, Wil17b},
privacy amplification against quantum side information \cite{SGC22b, SGC23},
quantum channel simulation \cite{BSS+02, LD09, BDH+14},
convex split lemma \cite{ADJ17, AARJ65, Wil17b, CG23, CGB23},
and distributed quantum source simulation \cite[Theorem 8.13]{hbook}, \cite{GHC23, GC24}.
Last year, the paper \cite{APW} studied the channel resolvability for quantum to quantum channels.

%% Statement of the problem
In this paper, we consider the optimal rate of channel resolvability, namely, the resolution rate required to approximate the channel output via optimal codes with asymptotically vanishing error.
This problem admits two settings.
The first one considers channel output being generated from a memoryless and stationary channel with i.i.d.~input from a fixed distribution, for which we call the \emph{fixed-input setting}; see Figure~\ref{fig:fixed-input} below.
The second one is the \emph{worst-input setting}, where we consider the outputs coming from a memoryless and stationary channel with the worst possible input distribution (possibly correlated non-i.i.d.); see Figure~\ref{fig:worst-input} below.
In other words, the worst-input setting concerns the channel output that
is the hardest to approximate.
Originally, the worst-input setting for classical channels was introduced for the strong converse of identification codes \cite{han1993approximation}.
That is, by deriving the direct part of the worst-input channel resolvability, the paper \cite{han1993approximation} showed the strong converse part of identification codes. At the same time, the strong converse of the worst-input channel resolvability was also obtained by the achievability of identification capacity.
An essential tool introduced in above analysis is the information spectrum method, which enables us to handle a general sequence of channels and sources.
The other setting regarding the fixed-input formulation was used for achieving strong secrecy for classical wiretap channel, which derives
the exponential decreasing rate for the amount of information leakage \cite{hayashi2006general2}
while the papers \cite{1377491,CWY,WNI} implicitly used similar techniques for CQ channels.
The strong converse and second order asymptotics for the fixed-input setting was shown in \cite{watanabe2014strong}.

%% Main contributions - Worst-input case
The main contributions of this paper is to establish the optimal rate of CQ channel resolvability for both the worst-input and the fixed-input settings.
For the worst-input setting, we show that the optimal rate is given by the maximum mutual information between the CQ channel input-output over all input distributions.
This quantity coincides with the channel capacity of message transmission over the CQ channel \cite{Hol98, ON99, Win99b} and also with its identification capacity \cite{loeber, ahlswede2002strong}.
We establish the result by employing two key tools: CQ channel version of information spectrum \cite{1207373} for the direct part, and the identification codes for CQ channel for its strong converse.

%% Main contributions - Fixed-input case
For the fixed-input setting, we show that the optimal rate is given by the minimum mutual information between the CQ channel input-output over all input distributions yielding the same target output state.
Here we focus on the strong converse, as the direct part can be shown via the classical-quantum soft covering lemma \cite{hayashi2012quantum,cheng2023error, SGC22b}.
Our strategy is an extension of the method by the paper \cite{watanabe2014strong} for the classical case, where the key idea is a separation of good codewords and bad codewords, classified based on the empirical distribution. Nevertheless, a direct adaption to quantum setting does not works. The main issue is the failure of Lemma 2 in \cite{watanabe2014strong}, which states if a codeword has the empirical output distribution away from the target distribution (hence a bad codeword), then its typical set is disjoint of the typical set of the target distribution so that one distinguishes the good and bad codes by some natural test. Such a nice separation does not exist in quantum because the output quantum states have different eigenbasis to the target state, then one cannot easily conclude the orthogonality of typical subspaces based on the separation of codewords.

We note that the classical approach in \cite{watanabe2014strong} implicitly uses the Sanov theorem, which states that in the i.i.d.~asymptotic setting,
the probability that the empirical distribution is different from the true distribution approaches zero exponentially.
However, a well-known quantum version of the Sanov theorem \cite{bjelakovic2005quantum,notzel2014hypothesis} does not discuss a quantum counterpart of empirical distributions.
Armed with the latest paper \cite{another} formulating
another quantum version of the Sanov theorem, we manage to adopt the quantum analogue of the strategy in \cite{watanabe2014strong}.
We remark that the strong converse for the fixed-input setting is more technical challenging as the method for the worst-input setting and the standard strong converse analysis for packing problems via data-processing inequality (see e.g.,~\cite{CG24}) do not apply here.

%% Comments on fully quanutum channel resolvability
Lastly, we explain why channel resolvability of a fully quantum channel does not cover that of a CQ channel.
An entanglement breaking (EB) channel is given by a measurement and state preparation process.
When using a projective measurement over a computation basis,
an EB channel becomes essentially a CQ channel for the aim of message transmission.
That is, a code for message transmission with such an EB channel
can be considered as a code for message transmission with the corresponding CQ channel.
However, this relation does not hold
for channel resolvability. This is because,
a pure input state for an EB channel
often works as a randomized input over a CQ channel.
That is, a channel resolvability code over an EB channel
cannot be simulated by a code over CQ channel of the same rate.
Therefore, the result in \cite{APW} for fully quantum channels does not directly apply to
channel resolvability for CQ channels.
We refer the readers to Remark~\ref{REM} for a more detailed discussion.

%% Organization
The remaining part of this paper is organized as follows.
In Section \ref{S2}, we first formulates the problem of
the worst-input and the fixed-input settings of channel resolvability for CQ channels.
Section \ref{S3} reviews the results for classical channels.
We show the direct part for the worst-input setting in Section \ref{S4}.
The strong converse of the fixed-input setting is proved in Section \ref{S5}.
Section \ref{S6} studies identification codes for
classical-quantum channels.
Section \ref{S7} is devoted to the strong converse of the worst-input setting.
We end the paper with a discussion in Section \ref{S8}.

{\bf Notations.} We write $\{A\le B\}$ for the support projection of the non-negative part of $B-A$. Given a state $\rho$, we sometimes use the short notation $\rho\{A\le B\}=\Tr(\rho \{A\le B\} )$.
The trace norm of a hermitian matrix is defined as $\|A\|_1 := \Tr( |A| )$. The logarithm in this paper is always with base $2$.

\section{Formulation of Channel Resolvability and our results}\Label{S2}
Let $\mathcal{X}$ be a finite alphabet corresponding to the input system.
and $\mathcal{H}$ be a finite-dimensional Hilbert space
corresponding to the output system, which is described by $B$.
We denote the set of density matrices over $\mathcal{H}$
by $\mathcal{S}(\mathcal{H})$.
Let $W:\mathcal{X}\to  \mathcal{S}(\mathcal{H}), x\mapsto W_x$ be a classical-quantum channel, where $\{W_x\}_{x\in \mathcal{X}}\subset \mathcal{S}(\mathcal{H})$ is the family of output states on $\mathcal{H}$. We denote by $\mathcal{P}(\mathcal{X})$ the set of probability distribution on $X$,
\begin{align*} \mathcal{P}(\mathcal{X})=\left\{p:\mathcal{X}\to [0,\infty)\ \Big|\  \sum_{x\in\mathcal{X}}p(x)=1
\right\}.
\end{align*}
Given an input distribution $p\in \mathcal{P}(\mathcal{X})$, we write
\[W(p):=\sum_{x \in \mathcal{X} }p(x)W_x\]
as the induced output state, and \begin{align*}
W \times p&:=\sum_{x} p(x)\ketbra{x}\ten W_x, \\
 \rho\times p&:=\sum_{x} p(x)\ketbra{x}\ten \rho,
\end{align*}
for the joint input-output states.
The key information quantities is quantum relative entropy
and mutual information, which are defined as follows.
Given two density matrices $\rho$ and $\sigma$ with $\text{supp}(\rho)\subseteq \text{supp}(\sigma)$,
the quantum relative entropy \cite{Ume54} is defined as
\begin{align}
D(\rho\|\sigma):=\Tr \rho (\log \rho-\log \sigma ).
\end{align}
The mutual information under the joint state $W \times p$
is given as
\begin{align}
I(X;B)_{W \times p}:=\sum_{x \in \mathcal{X}}p(x) D(W_x\|W(p)).
\end{align}

Let $M\in \mathbb{N}_+$ be a positive integer.
We define the \emph{optimal resolution error} of $M$-types to the input $p$ as
\begin{align}
	\label{eq:defn:error}
	\eps(p,W,M) :=\inf_{|\mathcal{C}|=M} \frac{1}{2}\norm{W_{\mathcal{C}}
-W(p)}{1},
\end{align}
where
the infimum is over all codebook $\mathcal{C}=(x_1,\cdots, x_M)\in \mathcal{X}^M$ of cardinality $M$ and $W_{\mathcal{C}}:=
\frac{1}{M}\sum_{m=1}^M W_{x_m}$.
Denote the set of $M$-types probability distributions %(resp. less than $M$)
\begin{align*} &\mathcal{P}_M(\mathcal{X})=\{ p\in \mathcal{P}(\mathcal{X})\ |\  M p(x)\in \mathbb{Z} \text{ for all } x\in \mathcal{X} \}.%\\ &\mathcal{P}_{\le M}(X)=\bigcup_{1\le m\le M}\mathcal{P}_m(X)
\end{align*}
The optimal resolution error can be equivalently formulated as
 \begin{align*}&\eps(p,W,M)=\inf_{q\in \mathcal{P}_M(\mathcal{X})}\frac{1}{2}\norm{W(p)-W(q)}{1},\\
   &\eps(W,M):=\sup_{p\in \mathcal{P}(\mathcal{X})}\inf_{q\in \mathcal{P}_M(\mathcal{X})}\frac{1}{2}\norm{W(p)-W(q)}{1},
 \end{align*}
where $\eps(W,M)$ is the \emph{worst case error} over all input distributions.

In the i.i.d.~$n$-shot setting, a memoryless and stationary channel maps each $n$-letter input $x^n \in \mathcal{X}^n$ to an output product state, namely,
\begin{align*}
W^{\otimes n}: x_1 x_2\ldots x_n \mapsto
W^{\otimes n}_{x^n}:=
W_{x_1} \otimes W_{x_2} \otimes \cdots \otimes W_{x_n}.
\end{align*}
Then, we define the corresponding optimal resolution error as
\begin{align*} &\eps(p^n,W^{\otimes n},M)=\inf_{q^n \in \mathcal{P}_M(\mathcal{X}^n)}\frac{1}{2}\norm{W^{\otimes n}(p^n)-W^{\otimes n}(q^n)}{1},
	\\
&\eps(W^{\otimes n},M)=\sup_{p^n \in \mathcal{P}(\mathcal{X}^n)}\inf_{q^n \in \mathcal{P}_M(\mathcal{X}^n)}\frac{1}{2}\norm{W^{\otimes n}(p^n )-W^{\otimes n}(q^n)}{1}.
 \end{align*}
%Here, $W^{\otimes n}: x_1 x_2\ldots x_n \mapsto W_{x_1} \otimes W_{x_2} \otimes \cdots \otimes W_{x_n}$ is an i.i.d.~channel, and $p^{\otimes n} \in \mathcal{P}(\mathcal{X}^n)$ is an $n$-fold i.i.d.~distribution.

The \emph{resolvability rate} for a fixed input $p$ (resp.~the worst-input) are given by the minimum achievable rate such that the optimal resolution error vanishes, i.e.,~
\begin{align*}
R(p,W)&:=\inf \{R>0\ |\ \forall \gamma>0, \liminf_{n\to\infty} \eps(p^{\otimes n},W^{\otimes n},\lfloor2^{n(R+\gamma)}\rfloor)=0\},
\\
R(W)&:=\inf \{R>0\ |\  \forall \gamma>0, \liminf_{n\to\infty} \eps(W^{\otimes n}, \lfloor2^{n(R+\gamma)}\rfloor,n)=0\}.
 \end{align*}
The quantity $R(p,W)$ corresponds to the optimal rate of CQ channel resolvability for the fixed-input setting (see Figure~\ref{fig:fixed-input}), while $R(W)$ is the optimal rate for the worst-input setting (see Figure~\ref{fig:worst-input}).

\begin{figure}[h!]
	\includegraphics[width=1\linewidth]{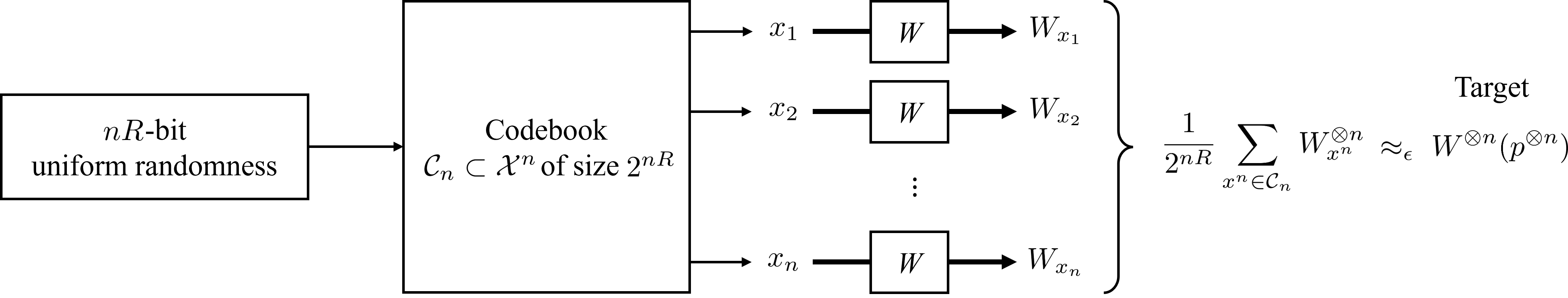}
	\centering
	\caption{A schematic diagram of channel resolvability for the fixed-input case.
		The goal is to employ $nR$-bit uniform randomness and an optimal codebook of size $2^{nR}$ to simulate the target state $W^{\otimes n}(p^{\otimes n})$, which is generated from a fixed i.i.d.~input distribution $p^{\otimes n}$.
		The minimum asymptotic rate $\liminf_{n\to \infty} \frac1n \log M$ for an $\epsilon$-approximation is $R_{\epsilon}(p,W)$ defined in \eqref{eq:def:fixed-input}.
	}
	\label{fig:fixed-input}
\end{figure}

\begin{figure}[h!]
	\includegraphics[width=1\linewidth]{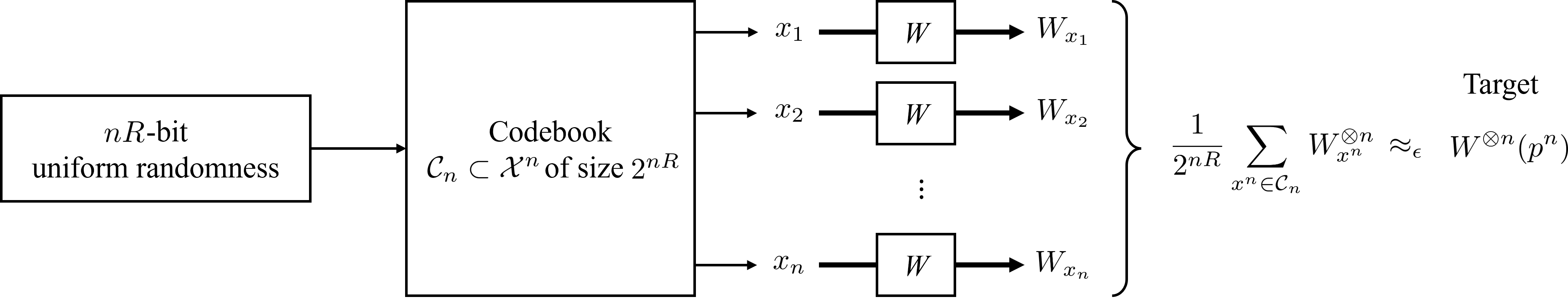}
	\centering
	\caption{A schematic diagram of channel resolvability for the worst-input case.
		The goal is to $nR$-bit uniform randomness and an optimal codebook of size $2^{nR}$ to simulate the target state $W^{\otimes n}(p^n)$, which is generated from a worst  (probably non-i.i.d.) input distribution $p^{n} \in \mathcal{P}(\mathcal{X}^n)$.
		The minimum asymptotic rate $\liminf_{n\to \infty} \frac1n \log M$ for an $\epsilon$-approximation is $R_{\epsilon}(W)$ defined in \eqref{eq:def:worst-input}.
	}
	\label{fig:worst-input}	
\end{figure}

 Our main results can be stated as follows
 \begin{align}
 &R(p,W)=\inf_{q:W(q)=W(p)}I(X;B)_{ W\times q},\Label{MA1}\\
 &R(W)=C(W):=\sup_{p \in \mathcal{P}(\mathcal{X}) }I(X;B)_{W\times p},\Label{MA2}
 \end{align}
 where $C(W)$ is channel capacity of $W$, and the infimum in  \eqref{MA1} is over all input distribution $q$ whose output $W(q)$ coincides with the target state $W(p)$.
 Moreover, we have the strong converse,
 \begin{align}
 &\liminf_{n\to \infty }\eps(p^{\otimes n}, W^{\otimes n},\lfloor2^{nR}\rfloor)\to 1\pl, \pl \text{ if } R < R(p,W), \Label{MA3}\\
 &\liminf_{n\to \infty } \eps(W^{\otimes n},\lfloor2^{nR}\rfloor)\to 1 \pl, \pl \text{ if } R < R(W).\Label{MA4}
 \end{align}
To state the strong converse in another word,
for fixed $\epsilon\in (0,1)$, we define
\begin{align}
	R_{\epsilon}(p,W):= \liminf_{n\to \infty}\frac{1}{n}\log M_\epsilon(p^{\otimes n},W^{\otimes n})\pl ,\quad   &M_\epsilon(p,W):= \min\{M \in \mathbb{N}_+ \ |\ \eps(p,W,M) \le \epsilon\},
	\label{eq:def:fixed-input}
	\\
	R_{\epsilon}(W):= \liminf_{n\to \infty}\frac{1}{n}\log M_\epsilon(W^{\otimes n})\pl ,\quad  &M_\epsilon(W):= \min\{M \in \mathbb{N}_+ \ |\ \eps(W,M) \le \epsilon\}.
	\label{eq:def:worst-input}
\end{align}
Clearly, for all $0< \epsilon <1$, we have
\begin{align}
R_{\epsilon}(W) &\le R(W),\Label{PSE}\\
R_{\epsilon}(p,W) &\le R(p,W)\Label{PSE2}.
\end{align}
Therefore, the strong converse parts \eqref{MA3} and \eqref{MA4}
can be restated as the equalities in \eqref{PSE} and \eqref{PSE2},
which can be derived by showing
\begin{align}
R_{\epsilon}(W) &\ge C(W),\Label{PSE3}\\
R_{\epsilon}(p,W) &\ge \inf_{q:W(q)=W(p)}I(X;B)_{ W\times q}\Label{PSE4}.
\end{align}
That is, the aim of the remaining part of the paper is to prove
the part $\le$ in \eqref{MA1},\eqref{MA2} (i.e.,~the direct part)
and to prove \eqref{PSE3}, \eqref{PSE4} (i.e.,~the strong converse). Note that here we do not consider the zero error resolvality $\epsilon=0$, which is different with the vanishing error case $R(W)$, and in most of case, an exact resolvability is not feasible for finite code length $n\in\mathbb{N}_+$.

The following theorems summarize our main results.

\begin{theorem}[Worst-Input Setting] \label{theorem:worst-input}
	For any classical-quantum channel $W:\mathcal{X} \to \mathcal{S(H)}$, the resolvability rate for the worst-input setting, defined in \eqref{MA2} and \eqref{eq:def:worst-input}, is given by
	\begin{align*}
		R(W)=R_{\epsilon}(W)
		=C(W):= \sup_{p \in \mathcal{P}(\mathcal{X}) }I(X;B)_{W\times p},
		\quad \forall \epsilon \in (0,1).
	\end{align*}
\end{theorem}

\begin{theorem}[Fixed-Input Setting] \label{theorem:fixed-input}
	For any classical-quantum channel $W:\mathcal{X} \to \mathcal{S(H)}$ and any input distribution $p\in\mathcal{P(X)}$,
	the resolvability rate for the fixed-input setting, defined in \eqref{MA1} and \eqref{eq:def:fixed-input}, is given by
	\begin{align*}
		R(p,W)=R_{\epsilon}(p,W)
		= \inf_{q:W(q)=W(p)}I(X;B)_{ W\times q},
		\quad \forall \epsilon \in (0,1).
	\end{align*}
Moreover, for $R<R(p,W)$, we have $\displaystyle\liminf_{n\to \infty }\eps(p^{\otimes n}, W^{\otimes n},\lfloor2^{nR}\rfloor)\to 1$ exponentially.
\end{theorem}

\section{Classical Channel Resolvability}\Label{S3}
For classical channels, i.e., $W:x\mapsto W_x$ with $W_x\in
\mathcal{P}(\mathcal{Y})$ being probability distributions over some finite alphabet $\mathcal{Y}$, the worst-input channel resolvability rate was shown to coincide with Shannon capacity  (see \cite{han1993approximation,verdu1994general} and also \cite{hayashi2006general2})
\begin{align} R(W)= C(W)=\sup_{p \in \mathcal{P(X)} } I(X;Y)_{ W\times p}\Label{eq:shannon}\end{align}
where $I(X;Y)= \sum_{x \in \mathcal{X}}p(x) D(W_x\|W(p)) $ is the mutual information of the joint input-output distribution $W\times p$. Later,
the reference \cite{watanabe2014strong} proved that for a fixed input $p \in \mathcal{P(X)}$,
\begin{align} R(p,W)= \inf_{q:W(q)=W(p)} I(X;Y)_{W\times q}. \Label{eq:fixed}\end{align}
Here, the infimum is over all possible input $q$ such that $W(q)=W(p)$. Such infimum is needed as the resolution error only depends the output $W(p)$ but not directly depend on the input $p$.
Moreover, the reference \cite{watanabe2014strong}
actually proved the strong converse that for any $R<R(p,W)$,
 \[\lim_{n\to\infty} 1-\eps(p, W^{\otimes n}, \lfloor2^{nR}\rfloor)=0 \text{ exponentially }.\]

For the relation between $R(p,W)$ and $R(W)$, one has
\begin{align}
\sup_{p \in \mathcal{P(X)}} R(p,W) =\sup_{p \in \mathcal{P(X)} } \inf_{q:W(q)=W(p)} I(X;Y)_{  W\times q}
\le \sup_{p \in \mathcal{P(X)} } I(X;Y)_{W\times p}=R(W).\Label{NBSA}
\end{align}

Here we emphasis that (the converse of) the fixed-input case \eqref{eq:fixed} does not imply (the converse of) the worst case \eqref{eq:shannon}. One can consider the case that there are two distributions $p$ and $q$ such that
\[W(p)=W(q), \quad I(X;Y)_{W\times q}<I(X;Y)_{ W\times p}=C(W).\]
which gives a strict inequality $\sup_{p\in\mathcal{P(X)}} R(p,W)<R(W)$ in \eqref{NBSA}.
Hence, even taking maximum of fixed-input rate $R(p,W)$ over all inputs $p\in\mathcal{P(X)}$, one cannot recover
the channel resolvability rate $R(W)$ for the worst input \eqref{eq:shannon}.
For this reason, we cannot say that
the converse for $R(W)$ is an easier problem than
the converse for $R(p,W)$.

\begin{exam} \label{example:separation}
	{\rm Consider the classical channel $W:\{0,1,\mathsf{e}\}\to \{0,1\}$.\\
		\begin{align} \label{table:separation}
			\begin{tabular}{|c|c|c|c|}
			\hline
			\diagbox{$Y$}{$X$} & $0$   & $1$   & $\mathsf{e}$
			\\
			\hline
			$0$ & $1-\epsilon$ & $\epsilon$ & $1/2$
			\\
			\hline
			$1$ & $\epsilon$ & $1-\epsilon$ & $1/2$
			\\
			\hline
		\end{tabular}
		\end{align}
		Given an input distribution $p=(1/2,1/2,0)$ that puts $50\%-50\%$ weights on letter $0$ and $1$,
		we calculate $I(X;Y)_{W\times p} = 1 - h(\epsilon)$, where $h(\epsilon) := \epsilon \log \frac{1}{\epsilon} + (1-\epsilon )\log \frac{1}{1-\epsilon}$ is the binary entropy function on $(\epsilon, 1-\epsilon)$.
		On the other hand,
		\begin{align}
			\notag
			D(W_0 \Vert W(p)) = D(W_1 \Vert W(p)) &= 1 - h(\epsilon),
			\\
			\notag
			D(W_{\mathsf{e}} \Vert W(p)) &= 0.
		\end{align}
		Hence, the channel capacity of $W$ is given by $C(W) = 1 - h(\epsilon)$.
		
		The capacity-achieving output distribution $W(p) = (1/2, 1/2)$ can be exactly simulated by another input distribution $q = (0,0,1)$ putting all weights on the letter $\mathsf{e}$, i.e.,~$W(q) = W(p) = (1/2, 1/2)$.
		In this case,
		$ I(X;Y)_{W\times q} = 0$.
		This showcases a strict separation for the inequality in  \eqref{NBSA} unless $\epsilon = 1/2$ (for which $D(W_0 \Vert W(p)) = D(W_1 \Vert W(p)) =
		D(W_{\mathsf{e}} \Vert W(p)) = 0$).
		The readers are referred to Figure~\ref{fig:separation} for an illustration.
		
		\begin{figure}[h!]
			\includegraphics[width=0.6\linewidth]{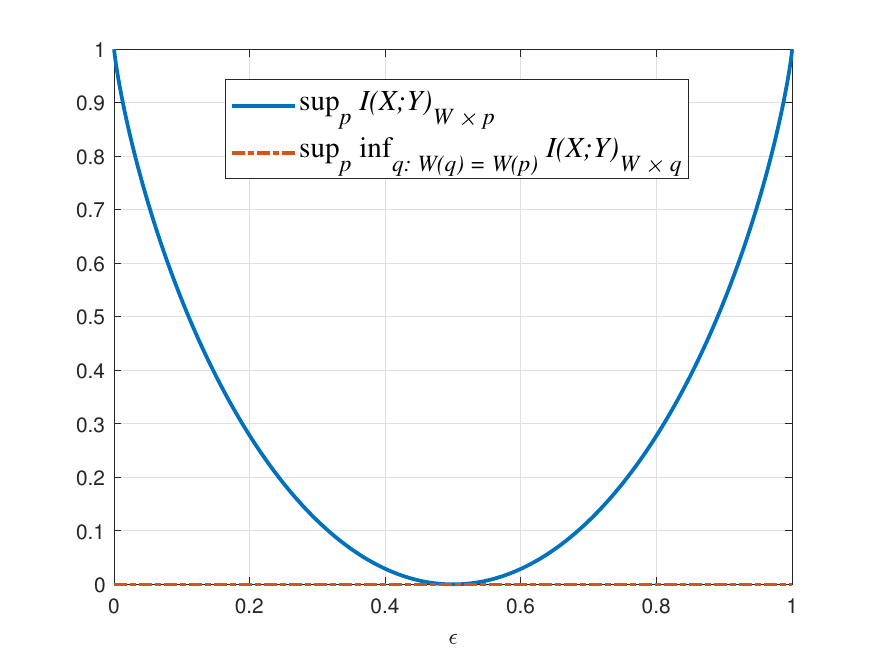}
			\centering
			\label{fig:separation}
			\caption{An illustration of the separation for \eqref{NBSA} via the channel given in \eqref{table:separation}.
			The upper curve is the channel capacity $1 + \epsilon \log \epsilon + (1-\epsilon )\log (1-\epsilon)$;
			the lower dotted curve is always $0$.
			The two curves coincide only when $\epsilon = 1/2$.
			}
		\end{figure}

	}
\end{exam}

%% Li's original example
%\begin{exam}{\rm Consider the classical channel $W:\{0,1,2\}\to \{0,1\}$.\\
%\[\begin{tabular}{|l|l|l|l|}
%\hline
%\diagbox{$Y$}{$X$} & 0   & 1   & 2     \\ \hline
%0 & 1/2 & 1/3 & 2/3  \\ \hline
%1 & 1/2 & 2/3 & 1/3   \\ \hline
%\end{tabular}\]
%Given the input distribution $p=(p_0,p_1,p_2)$, we have the output entropy
%\[ H(W(p))=h(\frac{1}{2}-\frac{1}{6}(p_2-p_1))\]
%only depending on the difference $p_2-p_1$. But the mutual information
%\[ I(X;Y)_{W\times p}=H(p)+H(W(p))-H(W\times p)= p_0\log\frac{3}{2}+\frac{2}{3}(p_1+p_2)\log 2 +\log 3+h(\frac{1}{2}-\frac{1}{6}(p_2-p_1))\]
%also depending on $p_1+p_2$. Suppose the optimal input for $R(W)$ is $p=(p_0,p_1,p_2)$ with $\delta=p_2-p_1$. One take
%\[q_1=(1-\delta, 0,\delta) ,\pl  q_2=(0,\frac{1-\delta}{2},\frac{1+\delta}{2})\pl.\]
%Since $\frac{2}{3}\log 2> \log \frac{3}{2}$, we have $I(X;Y)_{W\times q_1}<I(X;Y)_{W\times q_2}\le I(X;Y)_{W\times p}$.
%}
%\end{exam}

\section{Direct part for CQ-Channel Resolvability}\Label{S4}

The direct part of the fixed-input setting (which is also called the achievability part) follows from random coding strategies studied in \cite[§ IV]{hayashi2012quantum}, \cite[§ III]{cheng2023error}, \cite[§ IV]{SGC22b} under the name \emph{quantum soft covering}.
The idea is to employ a random codebook $\mathcal{C} = \{ x_1, x_2, \ldots, x_M  \}$ of size $M$, where each codeword $x_m$ is pairwise-independently drawn from a distribution $q \in \mathcal{P(X)}$ satisfying $W(q) = W(p)$.
Then, one can relax the optimal resolution error introduced in \eqref{eq:defn:error} as follows:
\begin{align*}
	\eps(p,W,M) \leq \frac{1}{2}\mathds{E}_{\mathcal{C} \sim q } \norm{W_{\mathcal{C}}
		-W(q)}{1}.
\end{align*}

It was proved in \cite[Theorem 1]{cheng2023error} that for any $\al\in (1,2]$,
\begin{align}\label{eq:sc}
	\frac{1}{2}\mathds{E}_{\mathcal{C} \sim q }\left\|W_{\mathcal{C}}
		-W(q)\right\|_{1}\le 2^{\frac{2}{\al}-2} 2^{\frac{\al-1}{\al}I_\al(X;B)_{W\times q}-\log M},
\end{align}
where $I_\al(X;B)_{W\times q} :=\inf_{\sigma} \frac{1}{\al-1}\log(\sum_{x}q(x)\| \sigma^{\frac{1-\al}{2\al}} W_x \sigma^{\frac{1-\al}{2\al}} \|_{\al}^\al)$ is the sandwiched R\'enyi mutual information \cite{MDS+13, WWY14, HT14, CGH18}.
Since $\displaystyle\lim_{\al\to 1}I_\al(X;B)_{W\times q}=I(X;B)_{W\times q}$, \eqref{eq:sc} implies (see also \cite[Lemma 10]{hayashi2012quantum} and \cite[Proposition 14]{SGC22b}) that,  if $R > I(X;B)_{W\times q}$,
\[\eps(p^{\otimes n}, W^{\otimes n}, \lceil2^{nR}\rceil)\le 2^{-n\sup_{\al\in (1,2]}\frac{\al-1}{\al}(R-I_\al(X;B)_{W\times q})}\to 0. \]
Minimizing all input distribution $q$ satisfying $W(q) = W(p)$ gives the direct part for the fixed-input setting, i.e.,
the part of $\le$ in \eqref{MA1} and Theorem~\ref{theorem:fixed-input}.

The remaining part of this section studies the part of $\le$ in \eqref{MA2} and Theorem~\ref{theorem:worst-input}.
In the worst-input setting, the analysis for $R(W)$ is more involved because
it requires the analysis for non-i.i.d.~inputs.
Moreover, as illustrated in \eqref{NBSA} and Example~\ref{example:separation}, there are channels $W$ such that $\sup_p \inf_{q:W(q)=W(p)} I(X;B)_{W\times q} < \sup_p I(X;B)_{W\times p}$.
Hence, the direct part of the fixed-input setting, i.e., $R(p,W) \leq \inf_{q:W(q)=W(p)} I(X;B)_{W\times q}$ is \emph{not} sufficient to prove the direct part of the worst-input setting.

To solve the worst-input case in the classical case, the paper \cite{hayashi2006general2} employs
its Eq.~(19) for $R(p,W)$, and its Eq.~(20) for $R(W)$.
Both are one-shot error bounds for classical-quantum channel communication.
The soft covering lemma employed in \cite[Lemma 10]{hayashi2012quantum}, \cite[Theorem 1]{cheng2023error}, and \cite[(27)]{SGC22b}
can be considered as quantum extensions of \cite[Eq.~(19)]{hayashi2006general2}.
Here, we develop a quantum extension of \cite[Eq.~(20)]{hayashi2006general2}.
We first recall the following lemma from the book \cite{hbook}.

\begin{lemma}[\protect{\cite[Lemma 9.2]{hbook}}]\Label{LL2}
For any state $\sigma$, we have
\begin{align}
&\min_{q\in \mathcal{P}_{ M}(X)}
\|W(p)-W(q) \|_1 \nonumber\\
&\le  4 \sqrt{\sum_{x}p(x) \Tr W_x
\{ \mathcal{E}_{\sigma} (W_x)\ge C \sigma \}
}+ \sqrt{\frac{v'}{M}\sum_x p(x) \Tr \sigma^{-1} \mathcal{E}_\sigma (W_x)^2
\{\mathcal{E}_{\sigma} (W_x)< C \sigma \}},
\end{align}
where $\mathcal{E}_{\sigma}$ is the pinching map of the spectrum of $\sigma$ and $v'$ is the number of distinct eigenvalues of $\sigma$.
\end{lemma}

To use above lemma, we prepare several notations as follows. Let $\rho$ be a state with spectral decomposition as
$\rho=\sum_{j}s_j E_j$,
where $s_1>s_2>\cdots >s_w$ are its distinct eigenvalues in decreasing order. We write $s_{\max}(\rho)=s_1=\|\rho\|$ as the maximum eigenvalue of $\rho$ and $s_{\min}(\rho)$ be the minimum eigenvalue of $\rho$. Based on the idea in
\cite[Theorem 24]{6872563} and \cite[Lemmas 8 and 9]{hayashi2012quantum},
for positive $\lambda>0$ and positive integer $v>0$, we define the positive operator as the spectral calculus of $\rho$
\[\lceil \rho \rceil_{\lambda,v}:=
\sum_{j}  s_1
2^{\lambda \cdot
\max(-v,\lceil \frac{\log s_j/s_1}{\lambda} \rceil)}
E_j =f_{v,\lambda}(\rho), \]
where $f_{v,\lambda}$ is the piecewise function
\[ f_{v,\lambda}(t)=\begin{cases}
                   s_12^{-k\lambda}, & \mbox{if }  (k-1)\lambda < \log (s_1/t)\le k\lambda \ \& \  t> s_1 2^{\lfloor-v\rfloor\lambda} \\
                   s_12^{-v\lambda}, & \mbox{if } t\le s_1 2^{\lfloor-v\rfloor\lambda}.
                 \end{cases}\]
We note that $t<f_{v,\lambda}(t)\le t2^{\lambda}+2^{-v\lambda }$.
Based on Lemma \ref{LL2}, we have the following quantum extension of \cite[Eq. (19)]{hayashi2006general2}.

\begin{lemma}\Label{LL1B}
We have
%Let $v$ be $-\lceil \frac{\log s_{\min}(W(p))/s_{\max}(W(p))}{\lambda} \rceil $.
\begin{align}
\min_{q\in \mathcal{P}_{ M}(X)}
\|W(p)-W(q) \|_1
\le 4 \sqrt{\sum_{x}p(x) \Tr W_x
\{ \mathcal{E}_{\lceil W(p)\rceil_{\lambda,v}} (W_x)\ge
L \lceil W(p)\rceil_{\lambda,v} \}
}+ \sqrt{\frac{v L}{M}}.\Label{NMZ}
\end{align}
\end{lemma}

\begin{proof}Let the dimension of output Hilbert space $\cH$ be $d$.
We choose $\sigma$ in Lemma \ref{LL2} to be $\lceil W(p)\rceil_{\lambda,v}
/\Tr \lceil W(p)\rceil_{\lambda,v}$, and $C$ to be
$L/ \Tr \lceil W(p)\rceil_{\lambda,v}$.
Then, we have
$$\{ \mathcal{E}_{\sigma} (W_x)\ge C \sigma \}=\{ \mathcal{E}_{\lceil W(p)\rceil_{\lambda,v}} (W_x)\ge
L \lceil W(p)\rceil_{\lambda,v} \}.$$
Since $$1\le \Tr \lceil W(p)\rceil_{\lambda,v} \le
\Tr( W(p) 2^\lambda
+2^{-v \lambda}I)
= 2^\lambda + 2^{-v \lambda} d,$$ we have
$$L (2^\lambda + 2^{-v \lambda} |d|)^{-1}
\le C\le L. $$
Then, for the second term in Lemma \ref{LL2}, we have
\begin{align}
&\sum_x p(x) \Tr \sigma^{-1} \mathcal{E}_\sigma (W_x)^2
\{\mathcal{E}_{\sigma} (W_x)< C \sigma \} \notag\\
\le &
\sum_x p(x) \Tr \mathcal{E}_\sigma (W_x) C
\{\mathcal{E}_{\sigma} (W_x)< C \sigma \} \notag\\
\le &
\sum_x p(x) \Tr \mathcal{E}_\sigma (W_x) C
=
 C \le   L .
\end{align}
We obtain \eqref{NMZ}.
\end{proof}

To show the direct part for $R(W)$,
i.e., the part of $\le$ in \eqref{MA2} and Theorem~\ref{theorem:worst-input},
we have to handle states that
do not have an i.i.d.~ form.
For this aim, we introduce the following notations for a general sequence of quantum states, which are originally formulated in the papers
\cite{1207373,4069150}. Note that the original statements from \cite{1207373,4069150} are of logarithm and exponential with base by the natural number $e$. Here we convert them to base $2$ to match with our discussion.

Given a sequence $\vec{p}$ of distributions $p_n$ on $\mathcal{X}^n$,
according to the paper \cite{1207373},
for $0 < \epsilon \le 1$,
we define
\begin{align}
\overline{I}^\dagger(\epsilon \,|\, \vec{p}, W^{\otimes n})
:=&
\inf \left\{ a \,\middle|\, \liminf_{n\to \infty}
\sum_{x^n \in \mathcal{X}^n} p_n(x^n) W_{x^n}^{\otimes n} \{
W_{x^n}^{\otimes n} \le 2^{na} W^{\otimes n}(p_n) \} \ge \epsilon
\right\} \notag\\
=&
\inf \left\{ a \,\middle|\, \limsup_{n\to \infty}
\sum_{x^n \in \mathcal{X}^n} p_n(x^n) W_{x^n}^{\otimes n}\{
W_{x^n}^{\otimes n} > 2^{na} W^{\otimes n}(p_n) \} \le 1-\epsilon
\right\} \notag\\
=&
\inf \left\{ a \,\middle|\, \limsup_{n\to \infty}
W^{\otimes n} \times p_n \{
W^{\otimes n} \times p_n > 2^{na} W^{\otimes n}(p_n) \times p_n\} \le 1-\epsilon
\right\}\Label{EE1}.
\end{align}
This quantity is related to strong converse for channel coding.

As its modification, for $0 < \varepsilon \le 1$,
$\lambda>0$, and a sequence of positive integers $\vec{v}:=(v(n))$,
we define the following.
\begin{align}
\overline{I}_{\lambda,\vec{v}}^\dagger(\epsilon |\vec{p}, W^{\otimes n})
:=&
\inf \left\{ a \,\middle|\, \limsup_{n\to \infty}
\sum_{x^n \in \mathcal{X}^n} p_n(x^n) W_{x^n}^{\otimes n}\{
W_{x^n}^{\otimes n} > 2^{na} \lceil W^{\otimes n}(p_n)\rceil_{\lambda,v(n)} \} \le 1-\epsilon
\right\} \notag\\
=&
\inf \left\{ a \,\middle|\, \limsup_{n\to \infty}
W^{\otimes n}\times p_n\{
W^{\otimes n}\times p_n > 2^{na} \lceil W^{\otimes n}(p_n)\rceil_{\lambda,v(n)}\times p_n \} \le 1-\epsilon
\right\}\Label{EE2}.
\end{align}

For $0 < \epsilon \le 1$,
the paper \cite[Theorem 1]{1207373} showed that
\begin{align}\label{eq:cw}
\sup_{\vec{p}} \overline{I}^\dagger(\epsilon|\vec{p}, W^{\otimes n})= C(W).
\end{align}

According to the paper \cite{4069150},
for $0 < \epsilon \le 1$,
we consider the hypothesis-testing exponent
for two general sequences of states $\vec{\rho}:=(\rho_n)$ and $\vec{\sigma}:=
(\sigma_n)$ as
\begin{align}
B^\dagger(\epsilon| \vec{\rho}\|\vec{\sigma}):=
\min_{T_n}
\Big\{
\liminf_{n \to \infty} -\frac{1}{n}\log
\Tr \sigma_n T_n \Big|
\liminf_{n \to \infty}
\Tr \rho_n (I-T_n)\ge \epsilon\Big\}.
\end{align}
In this setting, $\sigma_n$ do not need to be normalized
while the states $\rho_n$ need to be normalized.
The paper \cite{4069150} defines
\begin{align}\label{eq:d}
D^\dagger(\epsilon| \vec{\rho}\|\vec{\sigma}):=
\inf\big\{a \big|
\liminf_{n\to \infty}
\rho_n \{ \rho_n \le 2^{na} \sigma_n \} \ge \epsilon
\big\},
\end{align}
and showed the following lemma.
\begin{lemma}[\protect{\cite[Theorem 1]{4069150}}]\Label{LL8}
For $0 < \epsilon \le 1$, we have
\begin{align}
D^\dagger(\epsilon| \vec{\rho}\|\vec{\sigma})=
B^\dagger(\epsilon| \vec{\rho}\|\vec{\sigma}).
\end{align}
\end{lemma}
We shall now prove the direct part of channel resolvability in the worst-input setting, i.e., the part of $\le$ in \eqref{MA2} and Theorem~\ref{theorem:worst-input}.
\begin{theorem}[Direct Part of Worst-Input Setting]\Label{Th1}
For any classical-quantum channel $W:\mathcal{X}\to \mathcal{S(H)}$, we have
%Assume that $\min_{x \in \mathcal{X}}s_{\min}(W_x)>0$.
\begin{align}
 R(W) \le C(W)
%\sup_{\vec{p}} \overline{I}_\lambda(\eps'|\vec{p}, (W^{\otimes n})).
\Label{XPA}.
\end{align}
\end{theorem}

\begin{proof}
Fix some $\lambda >0$. Let $\text{dim}{\cH}=d$.
We choose positive integers $v(n)$ such that
\begin{align}
\lim_{n \to \infty}
\frac{v(n) \lambda}{n}- \log d
&>
B^\dagger(\epsilon| (W^{\otimes n}\times p_n)\| (W^{\otimes n}(p_n) \times p_n) ),\Label{NMT}
\\
\lim_{n \to \infty}
\frac{\log v(n) }{n}&=0. \Label{NME}
\end{align}
Since $W^{\otimes n}(p_n) \le \lceil W^{\otimes n}(p_n)\rceil_{\lambda,v(n)} \le 2^\lambda W^{\otimes n}(p_n) +2^{-v(n) \lambda} I$,
for $0 \le \epsilon <1$,
we have the following relation
\begin{align}
B^\dagger(\epsilon| (W^{\otimes n}\times p_n)\| (W^{\otimes n}(p_n) \times p_n) )=
B^\dagger(\epsilon| (W^{\otimes n}\times p_n)\| (\lceil W^{\otimes n}(p_n)\rceil_{\lambda,v(n)} \times p_n) ).
\end{align}
This is because by assumption \eqref{NMT}, for any test $(T_n)$, there is a large enough $n_0$ such that for any $n\ge n_0$,
\[\Tr \left((\lceil W^{\otimes n}(p_n)\rceil_{\lambda,v(n)} \times p_n) T_n \right)\le 2^{v}\Tr ((W^{\otimes n}(p_n) \times p_n) T_n)+2^{-v(n) \lambda}d^n\le (2+2^{v})\Tr ((W^{\otimes n}(p_n) \times p_n) T_n) \]
which does not change the limit exponent value as $n\to \infty$. Since the application of the completely-positive and trace-preserving map $\mathcal{E}_{\lceil W^{\otimes n}(p_n)\rceil_{\lambda,v(n)}}$ decreases the performance of hypothesis testing,
we have
\begin{align}
&B^\dagger(\epsilon| (\mathcal{E}_{\lceil W^{\otimes n}(p_n)\rceil_{\lambda,v(n)}}(W^{\otimes n}\times p_n))\| (\lceil W^{\otimes n}(p_n)\rceil_{\lambda,v(n)} \times p_n) )
\nonumber \\ \le &B^\dagger(\epsilon| (W^{\otimes n}\times p_n)\| (\lceil W^{\otimes n}(p_n)\rceil_{\lambda,v(n)} \times p_n) ).
\end{align}
%Therefore, using \eqref{EE1}, \eqref{EE2}, and \eqref{BV1}, we have
%\begin{align}
%\overline{I}(\eps|\vec{p}, (W^{\otimes n}))
%=\overline{I}_\lambda(\eps|\vec{p}, (W^{\otimes n})).\Label{BV2}
%\end{align}
Combining with Lemma \ref{LL8}, \eqref{eq:cw} and \eqref{eq:d},
\begin{align}
& D^\dagger(\epsilon| (\mathcal{E}_{\lceil W^{\otimes n}(p_n)\rceil_{\lambda,v(n)}}(W^{\otimes n}\times p_n))\| (\lceil W^{\otimes n}(p_n)\rceil_{\lambda,v(n)} \times p_n) )
\nonumber\\ \le &
D^\dagger(\epsilon| (W^{\otimes n}\times p_n)\| (\lceil W^{\otimes n}(p_n)\rceil_{\lambda,v(n)} \times p_n) ) \nonumber\\
=&D^\dagger(\epsilon| (W^{\otimes n}\times p_n)\| (W^{\otimes n}(p_n) \times p_n) \nonumber
\\ =&C(W)
\label{XA1}
\end{align}
%Combining %\eqref{BV2},
%\eqref{EE1}, \eqref{BV1}, and \eqref{XA1}, we have
%\begin{align}
%D^\dagger(\epsilon| (\mathcal{E}_{\lceil W^{\otimes n}(p_n)\rceil_{\lambda,v(n)}}(W^{\otimes n}\times p_n))\| (\lceil W^{\otimes n}(p_n)\rceil_{\lambda,v(n)} \times p_n) )
%\le C(W)
%\Label{XA2}
%\end{align}

In the same way as \eqref{EE2},
we have
\begin{align}
&D^\dagger(1-\epsilon| (\mathcal{E}_{\lceil W^{\otimes n}(p_n)\rceil_{\lambda,v(n)}}(W^{\otimes n}\times p_n)\| (\lceil W^{\otimes n}(p_n)\rceil_{\lambda,v(n)} \times p_n) )\notag\\
=&
\inf \left\{ a \middle| \limsup_{n\to \infty}
\sum_{x^n \in \mathcal{X}^n} p_n(x^n) \mathcal{E}_{\lceil W^{\otimes n}(p_n)\rceil_{\lambda,v(n)}}
(W_{x^n}^n) \big\{
\mathcal{E}_{\lceil W^{\otimes n}(p_n)\rceil_{\lambda,v(n)}}(W_x^n) > 2^{na} \lceil W^{\otimes n}(p_n)\rceil_{\lambda,v(n)} \big\} \le \epsilon
\right\} .
\end{align}
%Choose $R>R'$ such that
%$$R>R'>
%\sup_{\vec{p}} D^\dagger(1| (\mathcal{E}_{\lceil W^{\otimes n}(p_n)\rceil_{\lambda,v(n)}}(W^{\otimes n}\times p_n))\| \lceil W^{\otimes n}(p_n)\rceil_{\lambda,v(n)} \times p_n ) $$.
Take $\epsilon=0$ for the vanishing error and choose

\begin{align*}\vec{p}:=(p_n)\ , \  p_n=\argmax_{p_n \in \mathcal{P}(X^n)}\inf_{q\in \mathcal{P}_{
\lceil 2^{nR'}\rceil }(X^n)}\frac{1}{2}\norm{W^{\otimes n}(p_n)-W^{\otimes n}(q_n)}{1}
,
\end{align*}
and $R,R'$ such that
\begin{align}R>R'>
\sup_{\vec{p}} D^\dagger(1|( \mathcal{E}_{\lceil W^{\otimes n}(p_n)\rceil_{\lambda,v(n)}}(W^{\otimes n}\times p_n))\| (\lceil W^{\otimes n}(p_n)\rceil_{\lambda,v(n)} \times p_n )) .\label{eq:RR}\end{align}
Since the condition \eqref{NME} implies
\begin{align}
\sqrt{\frac{v(n)2^{nR'}}
{2^{nR}}}\to 0,
\end{align}
Lemma \ref{LL1B} with $L=\lfloor 2^{nR'}\rfloor$, $M=\lceil2^{nR}\rceil$
guarantees
\begin{align}
\max_{p_n \in \mathcal{P}(X^n)}\inf_{q_n \in \mathcal{P}_{
2^{nR}}(X^n)}\frac{1}{2}\norm{W^{\otimes n}(p_n)-W^{\otimes n}(q_n)}{1}\to 0.
\end{align}
Since $R,R'$ are arbitrary real numbers satisfying
the condition \eqref{eq:RR}, together with \eqref{XA1}
we have
\begin{align}
R(W) \le
\sup_{\vec{p}} D^\dagger(1| (\mathcal{E}_{\lceil W^{\otimes n}(p_n)\rceil_{\lambda,v(n)}}(W^{\otimes n}\times p_n))\| (\lceil W^{\otimes n}(p_n)\rceil_{\lambda,v(n)} \times p_n) )\le C(W) \Label{ZXO}.
\end{align}
That finishes the proof.
\end{proof}

\section{Strong converse for Channel Resolvability with fixed input}
\Label{S5}
\subsection{Quantum Sanov theorem}\Label{S5-A}
%As explained in Section \ref{S3},
%we cannot say that
%the converse for $R(W)$ is an easier problem than
% the converse for $R(p,W)$. That is, these two problems have different difficulties.

The goal of this section is to show the
strong converse for channel resolvability with fixed-input, basically, the inequality \eqref{PSE4} $$R_{\epsilon}(p,W)
		\ge \inf_{q:W(q)=W(p)}I(X;B)_{ W\times q},
		\quad \forall \epsilon \in (0,1)$$
in  Theorem~\ref{theorem:fixed-input}.
For this aim, we employ an alternative quantum version of Sanov theorem by \cite{another}, which is different from quantum Sanov theorem by \cite{Notzel,Bjelakovic}.
The quantum extension by \cite{another}
is based on empirical distributions
in the same spirit as the original Sanov theorem in the classical system
\cite{Bucklew,DZ}.
To state the above type of quantum Sanov theorem,
we prepare several notations.

Let $\mathcal{B}=\{|v_j\rangle\}_{j=1}^d
$ be the basis that diagonalizes the state $W(p)$.
We assume that our observation is based on
the basis $\mathcal{B}$, and
when our system is $\mathcal{H}^{\otimes n}$,
our observation is given as a basis
$|v[x^n]\rangle:=|v_{x_1}, \ldots,v_{x_n}\rangle$
with a data $x^n=(x_1, \ldots,x_n)$.
In this case, we obtain the data $x^n$.
The state $\sum_{j=1}^n \frac{1}{n}|v_{x_j}\rangle \langle v_{x_j}|$
is called its empirical state under the basis $\mathcal{B}$,
and is denoted by $e_{\mathcal{B}}(x^n)$.
%$\rho_{x^n,\mathcal{B}}$.
When $\mathcal{H} $ is a classical system, we only  have state vectors like $|x^n\rangle:=|{x_1}, \ldots,{x_n}\rangle$.
The state $e_{\mathcal{B}}(x^n)$ is simplified to
$e(x^n)$ and can be considered as the empirical distribution.

In the $n$-fold tensor product case,
we denote the set of empirical states under the basis $\mathcal{B}$
by
$$\mathcal{S}_n[\mathcal{B}]=\left\{ \sum_{j=1}^n \frac{1}{n}|v_{x_j}\rangle \langle v_{x_j}|\, \Big| (x_1, \ldots,x_n)\in \mathcal{X}^n\right\} .$$
Given an empirical state $\rho\in \mathcal{S}_n[\mathcal{B}]$,
we consider the projection
$$T_{\rho,\mathcal{B}}^n:= \sum_{x^n:
e_{\mathcal{B}}(x^n)
=\rho }|v[x^n]\rangle \langle v[x^n]|.$$
When $\rho \notin \mathcal{S}_n[\mathcal{B}]$,
$T_{\rho,\mathcal{B}}^n$ is defined to be $0$.

Let $\frS_n$ be the permutation group over $[n]:=
\{1, \ldots, n\}$.
For $g\in \frS_n$, we denote the unitary operation for changing the order of the tensor product based on $g$
by $U_g$.
We define the twirling  $\mathcal{W}_{\frS_n}$ as
\begin{align}
\mathcal{W}_{\frS_n}(\rho):=
\frac{1}{n!}\sum_{g \in \frS_n}
U_g \rho U_g^\dagger.
\end{align}

Next,
similar to the papers \cite{H-01,H-02,KW01,CM,CHM,OW,HM02a,HM02b,Notzel,H-24,H-q-text,AISW},
we employ Schur duality of $\cH^{\otimes n}$.
A sequence of monotone-increasing non-negative integers
$\blambda=(\lambda_1,\ldots,\lambda_d)$ is called Young index.
Although many references \cite{H-q-text,Group1,GW} define Young index as
monotone-decreasing non-negative integers, here
we define it in the opposite way for notational convenience.
We denote the set of Young indices $\blambda$ with condition
$\sum_{j=1}^d\lambda_j=n$ by $Y_d^n$.
Also, the set of probability distributions $p=(p_j)_{j=1}^d$ with the condition
$p_1\le p_2 \le \ldots \le p_d$ by $\mathcal{P}_d$.
For any density $\rho$ on $\mathcal{H}$,
the eigenvalues of $\rho$ forms an element of $\mathcal{P}_d$, which
is denoted by $p(\rho)$.
Let $\mathcal{P}_d^n$ be the set of elements $p$ of $\mathcal{P}_d$ such that
$p_j$ is an integer times of $1/n$.
We define the majorization relation $\succ$ in two elements of
$\mathcal{P}_d$ and $Y_d^n$.
For two elements $p,p'$ of $\mathcal{P}_d$, we say that
$p \succ p'$ when
\begin{align}
\sum_{j=1}^k p_j \ge \sum_{j=1}^k  p_j'
\end{align}
for $k=1, \ldots,d-1$,
and the majorization relation $\succ$ for two elements of $Y_d^n$ is defined
in the same way.

As explained in \cite[Section 6.2]{H-q-text}, one has
\begin{align}
\cH^{\otimes n}=
\bigoplus_{\blambda \in Y_d^n}
\mathcal{U}_{\blambda} \otimes \mathcal{V}_{\blambda},
\end{align}
where
$\mathcal{U}_{\blambda}$ expresses the irreducible subspace of $\SU(d)$
and $\mathcal{V}_{\blambda}$ expresses the irreducible subspace of
the representation $\pi$ of the permutation group $\frS_n$.
Denote
$$d_{\blambda}:= \dim \mathcal{U}_{\blambda}\ ,\  d_{\blambda}':= \dim \mathcal{V}_{\blambda}.$$
As shown in \cite[(6.16)]{H-q-text}, the dimension $d_{\blambda}$ is upper bounded as
\begin{align}
d_{\blambda}\le (n+1)^{\frac{d(d-1)}{2}}.\Label{NMI}
\end{align}
We denote the projection to the subspace $\mathcal{U}_{\blambda} \otimes \mathcal{V}_{\blambda}$ by
$P_{\blambda}$.
For $\blambda \in Y_d^n$ and $\rho \in \mathcal{S}_n[\mathcal{B}] $,
we define the projection
\begin{align}
T_{\blambda,\rho,\mathcal{B}}^n
:= P_{\blambda}
T_{\rho,\mathcal{B}}^n .
\end{align}
Note that
$T_{\blambda,\rho,\mathcal{B}}^n \neq 0$ holds
if and only if
\begin{align}
\blambda \succ n p(\rho).\Label{ZXN}
\end{align}
Since $\blambda$ is the highest weight vector,
weight vector $n p(\rho)$ contained in the space $\mathcal{U}_{\blambda}$
needs to satisfy the condition \eqref{ZXN}.
We also define the set
\begin{align}
\mathcal{R}[\mathcal{B}] &:=\{
(p,\rho) \in \mathcal{P}_d \times \mathcal{S}[\mathcal{B}]\ |\
p \succ p(\rho)\}, \\
\mathcal{R}_n[\mathcal{B}] &:=\mathcal{R}[\mathcal{B}]\cap
(\mathcal{P}_d^n \times \mathcal{S}_n[\mathcal{B}]).
\end{align}
That is, the decomposition $\{T_{\blambda,\rho,\mathcal{B}}^n\}_{
(\frac{\blambda}{n},\rho)\in \mathcal{R}_n[\mathcal{B}] }$
is the measurement to get the empirical distribution based on the basis $\mathcal{B}$
and the Schur sampling \cite{KW01,CM,CHM,OW,HM02a,HM02b,AISW}.
The following subset of $\mathcal{R}[\mathcal{B}]$ plays a key role in our quantum analogue of Sanov theorem.
For $\rho \in \mathcal{S}[\mathcal{B}]$ and $r>0$, we define
\begin{align} \label{eq:S_rho_r}
S_{\rho,r} %\notag\\
:=
\big\{  (p',\rho')\in \mathcal{R}[\mathcal{B}]\ \big|\
D(\rho'\| \rho)+H(\rho')-H (p') \le r \big\} ,
\end{align}
and $S_{\rho,r}^c$ as its complement, where $H(\rho):= - \Tr(\rho \log \rho)$.

The property of the set $S_{\rho,r}$ is characterized as follows. The readers are referred to \cite{another} for its proof.
\begin{prop}[\protect{\cite[Lemma 1]{another}}]\Label{LL1}
For $\rho \in \mathcal{S}[\mathcal{B}]$ and $r>0$, we have
\begin{align}
\Tr \rho^{\otimes n}
\left(\sum_{(p',\rho') \in S_{\rho,r}^c\cap \mathcal{R}_n[\mathcal{B}]}T_{np',\rho',\mathcal{B}}^n \right)
\le (n+1)^{\frac{(d+4)(d-1)}{2}}
2^{-nr},\Label{IOY}
\end{align}
and hence
\begin{align}
\lim_{n\to \infty}-\frac{1}{n}\log \Tr \rho^{\otimes n}
\left(\sum_{(p',\rho') \in S_{\rho,r}^c\cap \mathcal{R}_n[\mathcal{B}]}T_{np',\rho',\mathcal{B}}^n \right)
=r.\Label{IBT3}
\end{align}
\end{prop}

To study this case, we recall the sandwiched R\'{e}nyi relative entropy \cite{MDS+13, WWY14}
\begin{align}
D_{1+s}(\rho\|\sigma)&:= \frac{\phi(-s|\rho\|\sigma )}{s}, \\
\phi(s|\rho\|\sigma )&:= \log \Tr (\sigma^{\frac{s}{2(1-s)}}\rho\sigma^{\frac{s}{2(1-s)}})^{1-s}.
\end{align}
The quantum relative entropy $D(\rho\|\sigma)$ is characterized as
\begin{align}
D(\rho\|\sigma)= \lim_{s\to 0} D_{1+s}(\rho\|\sigma).
\end{align}
We define the pinching map $\mathcal{E}_{\mathcal{B}}$ as
\begin{align}
\mathcal{E}_{\mathcal{B}}(X):= \sum_{\rho \in \mathcal{S}_n[\mathcal{B}]}T_{\rho,\mathcal{B}}^n X
T_{\rho,\mathcal{B}}^n ,
\end{align}
where $X$ is an operator over $\mathcal{H}^{\otimes n}$.
The reference \cite[Next equation of (3.156)]{hbook} showed that
\begin{align}
-\phi(s|\mathcal{E}_{\mathcal{B}}(\sigma^{\otimes n})\|
\rho^{\otimes n} )
\le
-\phi(s|\sigma^{\otimes n}\|\rho^{\otimes n})
\le
-\phi(s|\mathcal{E}_{\mathcal{B}}(\sigma^{\otimes n})\|\rho^{\otimes n} )+s(d-1)\log (n+1).
\Label{XNF}
\end{align}
Then, we have the following version of quantum Sanov theorem.

\begin{prop}[\protect{\cite[Theorem 4]{another}}]\Label{TH1}
Let $\mathcal{B}$ be the eigenbasis of $\rho$.
For states $\rho_n \in \mathcal{S}_n[\mathcal{B}]$
and Young index $\blambda_n \in Y_d^n$
with the condition $\blambda \succ n p(\rho_n)$,
we have
\begin{align}
&%\liminf_{n\to \infty}
-\frac{1}{n}\log \Tr \sigma^{\otimes n}
T_{\blambda_n,\rho_n,\mathcal{B}}^n \notag\\
\ge &
\frac{-(1-s) r_n- \phi(1-s|\sigma \| \rho)}{s}
-\frac{d(d+3)}{2sn}\log (n+1) ,
\Label{ZPD2}
\end{align}
where $r_n:= D(\rho_n\| \rho)+H(\rho_n)-H (\frac{\blambda_n}{n})$.
\end{prop}

\subsection{Strong converse of fixed-input resolvability}
This subsection aims to prove the
strong converse for channel resolvability with fixed-input, i.e., the inequality \eqref{PSE4}, which is derived from the following theorem.
\begin{theorem}
Let
$p$ be an input distribution  satisfying the condition $
I(X;B)_{ W\times p}=\inf_{q:W(q)=W(p)}I(X;B)_{ W\times q}$.
Let $\mathcal{C}_n$ be a sequence of codebook $\mathcal{C}_n\subset X^n$ with size $|\mathcal{C}_n|=M_n$. Suppose that $\displaystyle R_1:= \liminf_{n\to \infty} \frac{1}{n}\log M_n< I(X;B)_{ W\times p}$.
Then,
\begin{align}
\limsup_{n\to \infty}\frac{1}{2}\|W^{\otimes n}_{\mathcal{C}_n}-W(p)^{\otimes n}\|_1
= 1.
\end{align}
\end{theorem}
We divide our proof into 6 steps.
The initial two steps are preparations for latter steps.

\noindent{\bf Step 1:} Throughout the proof, we fix two positive numbers $R_2$, $R_3$ such that $I(X;B)_{ W\times p}> R_3 >R_2 >R_1$.

In the first step, we will divide our codes into two subsets using information quantities.
For an element $x^n \in \mathcal{X}^n$,
we write $W[x^n]:= \frac{1}{n}\sum_{j=1}^n W_{x_j}$ as the empirical output, and define the set
\begin{align}
B_n(\delta,p):= \{x^n \in \mathcal{X}^n\ |\
\| W[x^n]-W(p) \|_1 \ge \delta\}.
\end{align}
$B_n(\delta,p)$ is the set of bad codewords whose empirical output states are away from $W(p)$.

Since $p$ satisfies $
I(X;B)_{ W\times p}=\inf_{q:W(q)=W(p)}I(X;B)_{ W\times q}$,
we can choose $\delta>0$ small enough to satisfy the following conditions:
for any element $x^n=(x_1, \ldots, x_n)\in B_n(\delta,p)$,
the empirical distribution $e(x^n)$ defined in the beginning of
Section \ref{S5-A} satisfies
\begin{align}
\frac{1}{n}\sum_{j=1}^n D(W_{x_j}\| W(p))
=I(X;B)_{ W\times e(x^n)}
> R_3>R_2\Label{NBI}.
\end{align}
%where $e(x^n)$ is the empirical distribution of $x^n$.
Furthermore, by the continuity of sandwiched R\'enyi relative entropy over the parameter $s\mapsto D_{1+s}$,
we choose $s_0 \in (0,1)$ such that
\begin{align}
\frac{1}{n}\sum_{j=1}^n-\frac{1}{s_0} \phi(s_0| W_{x_j} \|W(p))
=\frac{1}{n}\sum_{j=1}^n D_{1-s_0}(W_{x_j} \|W(p))
>  R_2 \Label{NHL}
\end{align}
for any element $x^n\in B_n(\delta,p)$.

We define the set
\begin{align}
\mathcal{S}(\delta,p):= \{ \sigma \in \mathcal{S}(\mathcal{H})\ |\
\| \sigma-W(p) \|_1 \ge \delta\}.
\end{align}
Then, for any $s \in (0,1)$, we have
\begin{align}
- \max_{\sigma \in \mathcal{S}(\delta,p)}\phi(1-s|\sigma \| W(p))
>0.
\end{align}
Given  the chosen $s_0 \in (0,1)$,
we choose $r>0$ sufficiently small such that
\begin{align}
%& \max_{s \in (0,1)}
%\frac{-(1-s) r- \max_{\sigma \in \mathcal{S}(\delta,p)}\phi(1-s|\sigma \| \rho)}{s} \\=&
R_4:= \frac{-(1-s_0) r- \max_{\sigma \in \mathcal{S}(\delta,p)}
\phi(1-s_0|\sigma \|W(p))}{s_0}
>0.\Label{EE27}
\end{align}

\noindent{\bf Step 2:}
In this proof, two projections play key roles.
The aim of this step is to introduce these two
projections, and to outline the structure of our proof.

We fix the basis $\mathcal{B}=\{|v_j\rangle\}_{j=1}^d
$ that diagonalizes the state $W(p)$.
Define the first projection:
\begin{align}
T_{r,n}:=
\sum_{(p_n,\rho_n) \in S_{W(p),r}\cap \mathcal{R}_n[\mathcal{B}]}T_{np_n,\rho_n,\mathcal{B}}^n .
\end{align}
For $\al>0$, we define
the second projection
$Q_n:=
\{\mathcal{E}_{\mathcal{B}}(W^{\otimes n}_{\mathcal{C}_n}) \ge 2^\alpha
T_{r,n}W(p)^{\otimes n} \}$.
In the latter part, we will evaluate the traces
$$\Tr W(p)^{\otimes n} T_{r,n} Q_n,
\Tr W(p)^{\otimes n} (I-T_{r,n}) Q_n, \text{ and }
\Tr W_{x^n}^n (I-Q_n),$$
separably.
In particular, the trace
$\Tr W_{x^n}^{\otimes n} (I-Q_n)$ will be evaluated depending on whether
$x^n \in \mathcal{C}_n \cap B_n(\delta,p)$, i.e.,
whether the codeword $x^n$ is bad or good.\\

\noindent{\bf Step 3:}
The aim of this step is to show the traces
$\Tr W(p)^{\otimes n} T_{r,n} Q_n$ and
$\Tr W(p)^{\otimes n} (I-T_{r,n}) Q_n$ are small.

Since $(I-T_{r,n})T_{r,n}W(p)^{\otimes n}=0$,
we have
\begin{align}
Q_n (I-T_{r,n})=(I-T_{r,n}). \Label{EE36}
\end{align}
Also,
\begin{align}
\Tr W(p)^{\otimes n} T_{r,n} Q_n
\le \Tr \mathcal{E}_{W^{\otimes n}(p^{\otimes n})}(W^{\otimes n}_{\mathcal{C}_n}) e^{-\alpha} Q_n
\le 2^{-\alpha}.\Label{EE37}
\end{align}
Using \eqref{IOY} of Proposition \ref{LL1}, we have
\begin{align}
\Tr W(p)^{\otimes n} (I-T_{r,n})Q_n \le
\Tr W(p)^{\otimes n} (I-T_{r,n})
\le \beta_{2,n} :=(n+1)^{\frac{(d+4)(d-1)}{2}} 2^{-nr}\Label{EE38}.
\end{align}

\noindent{\bf Step 4:}
The aim of this step is to show the trace
$\Tr W_{x^n}^{\otimes n} (I-Q_n)$ is small
for a bad codeword
$x^n \in \mathcal{C}_n \cap B_n(\delta,p)$.

Since $|S_{W(p),r}\cap \mathcal{R}_n[\mathcal{B}]|\le (n+1)^{2(d-1)}$,
using \eqref{ZPD2} of Proposition \ref{TH1} and \eqref{EE27},
we have for any state $\sigma$,
\begin{align}
\Tr \sigma^{\otimes n}
T_{r,n}
\le &
(n+1)^{2(d-1)}
(n+1)^{-\frac{d(d+3)}{2s_0}}
2^{-n\frac{-(1-s_0)r - \phi(1-s_0|\sigma \| W(p))}{s_0}} \notag\\
\le &
(n+1)^{2(d-1)}
(n+1)^{-\frac{d(d+3)}{2s_0}}
2^{-n R_4}. \Label{EE30}
\end{align}

For an element $x^n\in B_n(\delta,p)$,
the empirical state $e(x^n)$ defined in the beginning of
Section \ref{S5-A} satisfies
\begin{align}
\mathcal{W}_{\frS_n}(|x^n\rangle \langle x^n|)
\le (n+1)^{|\mathcal{X}|-1}e(x^n)^{\otimes n}.\Label{EE31}
\end{align}
Thus, we have
\begin{align}
&\Tr W_{x^n}^{\otimes n} T_{r,n}
=\Tr W_{x^n}^{\otimes n} \mathcal{W}_{\frS_n}(T_{r,n})
=\Tr \mathcal{W}_{\frS_n}(W_{x^n}^{\otimes n}) T_{r,n}
=\Tr W^{\otimes n}(\mathcal{W}_{\frS_n}(|x^n\rangle \langle x^n|)) T_{r,n}\notag\\
\stackrel{(a)}{\le}& (n+1)^{|\mathcal{X}|-1}
\Tr W^{\otimes n}(e(x^n)^{\otimes n}) T_{r,n}\notag\\
= & (n+1)^{|\mathcal{X}|-1}
\Tr W (e(x^n))^{\otimes n} T_{r,n}\notag\\
\stackrel{(b)}{\le}&  (n+1)^{|\mathcal{X}|-1}(n+1)^{2(d-1)}(n+1)^{-\frac{d(d+3)}{2s_0}}:=\beta_{1,n}
2^{-n R_4}.\Label{EE35}
\end{align}
Here, $(a)$ follows from \eqref{EE31},
$(b)$ follows from \eqref{EE30} and
$\norm{e(x^n)-W(p)}{1}>\delta $ by the assumption that $x^n\in B_n(\delta,p)$ is a bad codeword.\\

For a bad codeword $x^n \in \mathcal{C}_n \cap B_n(\delta,p)$,
using \eqref{EE35} and \eqref{EE36},
we have
\begin{align}
\Tr W_{x^n}^{\otimes n} Q_n
\ge \Tr W_{x^n}^{\otimes n} (I-T_{r,n})Q_n
=\Tr W_{x^n}^{\otimes n} (I-T_{r,n})
\ge 1- \beta_{1,n}. \Label{EE39}
\end{align}
Thus, we have
\begin{align}
\Tr W_{x^n}^{\otimes n} (I-Q_n) \le \beta_{1,n}.\Label{EE40}
\end{align}

\noindent{\bf Step 5:}
The aim of this step is the evaluation of
$\Tr W_{x^n}^{\otimes n} (I-Q_n)$
for a good codeword
$x^n \in \mathcal{C}_n \cap B_n(\delta,p)^c$.

Define the pinching $\mathcal{E}_{n}$ as
\begin{align}
\mathcal{E}_{n}(X):=Q_n X Q_n+(I-Q_n) X (I-Q_n).
\end{align}
Denote $M_n=|\mathcal{C}_n|$. Since
\begin{align}
M_n \mathcal{E}_{\mathcal{B}}(W^{\otimes n}_{\mathcal{C}_n})
\ge \mathcal{E}_{\mathcal{B}}(W^{\otimes n}_{x^n}),
\end{align}
and the application of the pinching $\mathcal{E}_{n}$
preserves the matrix inequality,
we have
\begin{align}
M_n \mathcal{E}_{\mathcal{B}}(W^{\otimes n}_{\mathcal{C}_n})
\ge \mathcal{E}_{n}\circ \mathcal{E}_{\mathcal{B}}(W^{\otimes n}_{x^n}).
\end{align}
Thus
\begin{align}
I-Q_n =\{\mathcal{E}_{\mathcal{B}}(W^{\otimes n}_{\mathcal{C}_n}) < 2^\alpha
T_{r,n}W(p)^{\otimes n} \}
%\subset
\le
\left\{ \frac{1}{M_n}
 \mathcal{E}_{n}\circ \mathcal{E}_{\mathcal{B}}(W^{\otimes n}_{x^n})
< 2^\alpha
T_{r,n}W(p)^{\otimes n} \right\}\Label{EE44}.
\end{align}

Because
$T_{r,n}$, $W(p)^{\otimes n}$, and $\mathcal{E}_{n}\circ \mathcal{E}_{\mathcal{B}}(W_{x^n}^{\otimes n} )$
are commutative with each other, we have
\begin{align}
&(\mathcal{E}_{n}\circ \mathcal{E}_{\mathcal{B}}(W_{x^n}^{\otimes n} ))^{-s}
(T_{r,n}W(p)^{\otimes n})^s \notag\\
\le &
(\mathcal{E}_{n}\circ \mathcal{E}_{\mathcal{B}}(W_{x^n}^{\otimes n} ))^{-s}
(W(p)^{\otimes n})^s \notag\\
=&
(W(p)^{\otimes n})^{s/2}
(\mathcal{E}_{n}\circ \mathcal{E}_{\mathcal{B}}(W_{x^n}^{\otimes n} ))^{-s}
(W(p)^{\otimes n})^{s/2} \notag\\
\stackrel{(c)}{\le}&
2^{-s}
(W(p)^{\otimes n})^{s/2}
(\mathcal{E}_{\mathcal{B}}(W_{x^n}^{\otimes n} ))^{-s}
(W(p)^{\otimes n})^{s/2} .\Label{EE48}
\end{align}
Here, $(c)$ follows from the pinching inequality
$\mathcal{E}_{\mathcal{B}}(W_{x^n}^{\otimes n} ) \le
2 \mathcal{E}_{n}\circ \mathcal{E}_{\mathcal{B}}(W_{x^n}^{\otimes n} )$.
Hence, we have
\begin{align}
&\Tr W_{x^n}^{\otimes n}(I- Q_n)
\\ \stackrel{(d)}{\le}&
\Tr W_{x^n}^{\otimes n}
\left\{ \frac{1}{M_n}
 \mathcal{E}_{n}\circ \mathcal{E}_{\mathcal{B}}(W^{\otimes n}_{x^n})
< 2^\alpha
T_{r,n}W(p)^{\otimes n} \right\}\notag\\
\stackrel{(e)}{\le}&
2^{s_0 \alpha} M_n^{s_0} \Tr W^{\otimes n}_{x^n}
(\mathcal{E}_{n}\circ \mathcal{E}_{\mathcal{B}}(W_{x^n}^{\otimes n} ))^{-s_0}
(T_{r,n}W(p)^{\otimes n})^{s_0} \notag\\
\stackrel{(f)}{\le}&
2^{-s_0} 2^{s_0 \alpha} M_n^{s_0} \Tr W^{\otimes n}_{x^n}
(W(p)^{\otimes n})^{s_0/2}
(\mathcal{E}_{\mathcal{B}}(W_{x^n}^{\otimes n} ))^{-s_0}
(W(p)^{\otimes n})^{s_0/2} \notag\\
= &
2^{-s_0} 2^{s_0 \alpha} M_n^{s_0} \Tr \mathcal{E}_{\mathcal{B}}(W^{\otimes n}_{x^n})
(W(p)^{\otimes n})^{s_0/2}
(\mathcal{E}_{\mathcal{B}}(W_{x^n}^{\otimes n} ))^{-s_0}
(W(p)^{\otimes n})^{s_0/2} \notag\\
= &
2^{-s_0} 2^{s_0 \alpha} M_n^{s_0}
\Tr (\mathcal{E}_{\mathcal{B}}(W_{x^n}^{\otimes n} ))^{1-s_0}
(W(p)^{\otimes n})^{s_0} \notag\\
= &
2^{-s_0} 2^{s_0 \alpha} M_n^{s_0} 2^{\phi(s_0| \mathcal{E}_{\mathcal{B}}(W_{x^n}^{\otimes n} )\|W(p)^{\otimes n})}\notag\\
\stackrel{(g)}{\le}&
(n+1)^{s_0(d-1)}
2^{-s_0} 2^{s_0 \alpha} M_n^{s_0} 2^{\phi(s_0| W_{x^n}^n \|W(p)^{\otimes n})}\notag\\
= & (n+1)^{s_0(d-1)}
2^{-s_0} 2^{s_0 \alpha} M_n^{s_0} 2^{\sum_{j=1}^n \phi(s_0| W_{x_j} \|W(p))}\notag\\
\stackrel{(h)}{\le}&
 (n+1)^{s_0(d-1)}
2^{-s_0} 2^{s_0 \alpha} M_n^{s_0} 2^{-n s_0R_2}=:\beta_{3,n}
.\Label{EE58}
\end{align}
Here, $(d)$ follows from \eqref{EE44},
$(e)$ follows from the fact that
the condition $$\frac{1}{M_n}
 \mathcal{E}_{n}\circ \mathcal{E}_{\mathcal{B}}(W^{\otimes n}_{x^n})
< 2^\alpha
T_{r,n}W(p)^{\otimes n} \Longrightarrow I \le
2^{s \alpha} M_n^{s_0}
(\mathcal{E}_{n}\circ \mathcal{E}_{\mathcal{B}}(W_{x^n}^{\otimes n} ))^{-s_0}
(T_{r,n}W(p)^{\otimes n})^{s_0},$$
$(f)$ follows from \eqref{EE48},
$(g)$ follows from \eqref{XNF}, and
$(h)$ follows from \eqref{NHL}.
By the choice of $R_2$, we have
\begin{align}
\limsup_{n\to \infty}-\frac{1}{n}\log\beta_{3,n}=s(R_2-R_1)>0 \Rightarrow  \liminf_{n\to \infty}\beta_{3,n}=0\Label{EE59}.
\end{align}

\noindent{\bf Step 6:} As the final step, we evaluate the desired trace norm by using
the above evaluations.
We have
\begin{align}
&1-\frac{1}{2}\|W^{\otimes n}_{\mathcal{C}_n}-W(p)^{\otimes n}\|_1\notag\\
\le &
1-\Tr (W^{\otimes n}_{\mathcal{C}_n}- W(p)^{\otimes n}) Q_n
\\ = &
\Tr W^{\otimes n}_{\mathcal{C}_n}(I-Q_n)
+
\Tr W(p)^{\otimes n} Q_n\notag\\
= &
\frac{1}{M_n}
\sum_{x^n \in \mathcal{C}_n \cap B_n(\delta,p)^c}
\Tr W_{x^n}^{\otimes n} (I-Q_n)
+
\frac{1}{M_n}
\sum_{x^n \in \mathcal{C}_n
 \cap B_n(\delta,p)}
\Tr W_{x^n}^{\otimes n} (I-Q_n) \notag\\
&+\Tr W(p)^{\otimes n} T_{r,n} Q_n
+\Tr W(p)^{\otimes n} (I-T_{r,n}) Q_n\notag\\
\stackrel{(i)}{\le}&
\frac{1}{M_n}
\sum_{x^n \in \mathcal{C}_n \cap B_n(\delta,p)^c}
\beta_{3,n}
+
\frac{1}{M_n}
\sum_{x^n \in \mathcal{C}_n \cap B_n(\delta,p)}
\beta_{1,n}
+
2^{-\alpha}
+\beta_{2,n} \notag\\
\le &
\max(\beta_{1,n},\beta_{3,n} )
+2^{-\alpha}
+\beta_{2,n} . \Label{EE65}
\end{align}
$(i)$ follows from \eqref{EE37}, \eqref{EE38}, \eqref{EE40}, and \eqref{EE58}.
Therefore, combining
\eqref{EE35}, \eqref{EE38}, \eqref{EE59}, and \eqref{EE65},
we have
\begin{align}
\liminf_{n\to \infty}
1-\frac{1}{2}\|W^{\otimes n}_{\mathcal{C}_n}-W(p)^{\otimes n}\|_1
\le 2^{-\alpha}.
\end{align}
Since $\alpha>0$ is arbitrary, we obtain
\begin{align}
\liminf_{n\to \infty}
1-\frac{1}{2}\|W^{\otimes n}_{\mathcal{C}_n}-W(p)^{\otimes n}\|_1
=0.
\end{align}

Although the paper \cite{watanabe2014strong} discussed the second-order asymptotics for classical channels,
our method does not work with
the second-order asymptotics for the CQ case in the following reason.
First, in our method, we apply the pinching map $\mathcal{E}_n$.
There is a possibility that
the pinching map $\mathcal{E}_n$ might cause an unexpected behavior of the probability
$\Tr W_{x^n}^{\otimes n}
\{ \frac{1}{M_n}
 \mathcal{E}_{n}\circ \mathcal{E}_{\mathcal{B}}(W^{\otimes n}_{x^n})
< 2^\alpha
T_{r,n}W(p)^{\otimes n} \}$.
%, which directly connects the second order asymptotics.
To avoid this problem,
we evaluate the R\'{e}nyi relative entropy, which enables us to show the strong converse.
However, it cannot derive the second-order asymptotics.

\begin{rem}\label{REM}{\rm
Here we discuss the relation with the approach by \cite{APW}.
The paper \cite{APW} studied the channel resolvability for fully quantum channels.
Theorem 31 of \cite{APW} gives a lower bound of
$R_{\epsilon}(p,W)$ for a quantum channel from system $A$ to the system $B$ by the smoothing  of conditional entropy $-H(A_R|B)$, where $A_R$ is the reference system of the input system $A$.

Nevertheless, the channel resolvability of a fully quantum channel does not cover that of a CQ channel. For the aim of message transmission, a CQ channel is equivalent to an
entanglement breaking (EB) channel
given by a projective measurement over the computation basis and a state preparation process. Nevertheless, this equivalence does not hold
for channel resolvability. When we treat a CQ channel as an entanglement breaking channel in  quantum to quantum channel setting, we are allowed to input a superposition pure state for the computation basis as a deterministic input. However, such an input works as a randomized input over a CQ channel because
the outcome of the projective measurement over the computation basis
is probabilistic.
Hence it cannot be considered as a deterministic input in the framework of CQ channel. Therefore, a quantum resolvability code over an EB channel
cannot be simulated by a resolvability code over CQ channel at the same rate. This is one of the difficulties of channel resolvability.

Indeed, for the entanglement-breaking channel,
the state over the joint system between
the output system $B$ and
the reference $A_R$ is a separable state.
Hence, the conditional entropy $H(A_R|B)$ is always
a non-negative value \cite{nielsen2001separable}. Hence, $-H(A_R|B)$ is not even a positive value.
Thus, the method \cite{APW} does not give an effective lower bound for
$R_{\epsilon}(p,W)$ for classical-quantum channels.
}

%In the entanglement-breaking channel, the state over the joint system between the output system $B$ and the reference $A_R$ is a separable state.
%Hence, the conditional entropy is a non-negative value.
%Hence, $-H(|)$ is not a positive value.
%That is, the method does not give a useful lower bound for  ... for classical-quantum channel.

\end{rem}

\section{Strong converse for channel resolvability of worst input}\Label{S6}
The goal of this section is to show the
strong converse for channel resolvability with worst-input, i.e., inequality \eqref{PSE3}
$$R_{\epsilon}(W) \ge C(W),$$
by using identification codes.
For classical channels, the channel resolvability was introduced to show the strong converse of the identification code, and
the converse of channel resolvability can be proved using the direct part of identification code. For CQ channels, the identification codes and its capacity have been studied in~\cite{loeber}.

\subsection{Identification codes}
The aim of this subsection is to summarize basic facts
for identification codes for CQ-channels, which works as
our preparation for strong converse of channel resolvability for the worst-input setting.

Let $W:x\mapsto W_x$ be a CQ channel and $0<\lambda_1,\lambda_2<1$.
An $(N,\lambda_1,\lambda_2)$ identification (ID) code $(p_i,D_i)_{i=1}^N$ of $W$ consists of a family of probability distribution $p_i\in \mathcal{P(X)}$ and corresponding measurement operator $D_i$ on $\mathcal{H}$ with $0\le D_i\le I$ such that for all $i,j=1,\cdots,N$ and $i\neq j$
\begin{align*}
\Tr(W(p_i)D_i)\ge 1-\lambda_1\ ,\
\Tr(W(p_i)D_j)\le \lambda_2.
\end{align*}
If in additional, there exists a common refinement of $D_i$, that is, there exists a POVM $\{E_j\}_{j=1}^k$, $0\le E_j\le I, \sum_{j}E_j=I$ such that
\[D_j=\sum_{j\in \mathcal{A}_j}E_j \text{ for some subset } \mathcal{A}_j\subset [k]\pl,\]
then $(p_i,D_i)_{i=1}^N$ is called a simultaneous $(N,\lambda_1,\lambda_2)$ ID code for $W$. In the $n$-shot setting, a $(N,\lambda_1,\lambda_2)$ code for the product channel $W^{\otimes n}$ is called a code $(n,N,\lambda_1,\lambda_2)$ for $W$.

Denote $N(n,\lambda_1,\lambda_2)$ (resp. $N_{\text{sim}}(n,\lambda_1,\lambda_2)$) as the maximum $N$ such that there exists a $(n,N,\lambda_1,\lambda_2)$ (resp. simultaneous) ID-code and the (resp. simultaneous) ID capacity is defined as
\begin{align*}
&C_{\lambda_1,\lambda_2}(W)=\liminf_{n\to \infty} \frac{\log\log N(n,\lambda_1,\lambda_2)}{n}\\
&C_{\text{sim},\lambda_1,\lambda_2}(W)=\liminf_{n\to \infty}
\frac{\log\log N_{\text{sim}}(n,\lambda_1,\lambda_2)}{n}.
\end{align*}
From the definition, we find that
\begin{align}
C_{\lambda_1,\lambda_2}(W) \ge
C_{\text{sim},\lambda_1,\lambda_2}(W).
\end{align}
It was proved by \cite{loeber, ahlswede2002strong} that for a CQ channel,the simultaneous identification capacity and the identification capacity coincides and are given by
\begin{align} {C}_{\lambda_1,\lambda_2}(W)={C}_{\text{sim},\lambda_1,\lambda_2}(W)=C(W)\Label{eq:idcapacity}\end{align}
for any $0<\lambda_{1},\lambda_2<1$.
%Here $\chi(W)$ is the Holevo quantity
%\[ C(W)=\sup_{p} I(X;B)_\rho\]

The connection between ID code and overall channel resolvability is provided by the following lemma, whose classical version is due to Han and Verd\'u \cite{han1993approximation}, and the quantum version is used by Loeber \cite{loeber} in his proof of achievability of ID code.
\begin{lemma}\Label{lemma:bridge}
Let $W:\mathcal{X}\to \mathcal{S}(\mathcal{H}), x\mapsto W_x$ be a classical quantum channel.
If there exists an $(N,\lambda_1,\lambda_2)$ identification code $\mathcal{C}=(p_i,D_i)_{i=1}^N$ for $W$ and an integer $M$ satisfy
\[1-\lambda_1-\lambda_2>\eps(W,M)\]
then
$$ |\mathcal{X}|^M\ge N.$$
\end{lemma}
\begin{proof}The proof is essentially in Lober's paper  \cite[Lemma 4.10 and Theorem 4.11]{loeber}. We repeat it here for completeness. Denote
$\mathcal{D}_i:\mathcal{S}(\mathcal{H})\to l_\infty(\{0,1\})$ as the binary measurement
\[\mathcal{D}_i(\rho)=(\Tr(\rho D_i), \Tr(\rho (I-D_i)) ).\]
Then for any $i\neq j$, we have
\[ \norm{W(p_i)-W(p_j)}{1}\ge \norm{\mathcal{D}_i\circ (W(p_i))-\mathcal{D}_i\circ W(p_j)}{1}\ge 2(1-\lambda_1-\lambda_2)>2\eps. \]
By the definition of $\eps(M,W)$, there exists $M$-type distribution $(q_1,\cdots,q_N)$ such that
\[ \frac{1}{2}\norm{W(p_i)-W(q_i)}{1}\le \eps. \]
In total, there are less than $|\mathcal{X}|^M$ many $M$-types in $P(\mathcal{X})$. Now, suppose $N>|\mathcal{X}|^M$, there must be some $i\neq j$ such that $q_i=q_j$. It follows that
\[ \frac{1}{2}\norm{W(p_i)-W(p_j)}{1}\le\norm{W(p_i)-W(q_i)}{1}+\norm{W(p_j)-W(q_j)}{1}\le 2\eps\]
which is a contradiction. This finishes the proof.
\end{proof}

The above lemma says if there exists an $(N,\lambda_1,\lambda_2)$ identification code $\mathcal{C}=(p_i,D_i)_{i=1}^N$ with
\[ N>|\mathcal{X}|^M,\]
then
\[ \eps(M,W)\le 1-\lambda_{1}-\lambda_2 .\]
Let us denote the optimal error for identification code as
\[\eps_{\text{ID}}(N,W):=\inf \{\lambda_{1}+\lambda_2\pl |\pl \exists (N,\lambda_1,\lambda_2) \text{ ID code}\pl\}.\]
Then the above lemma gives the following one shot estimate
between the minimum error of identification codes and
the minimum error of channel resolvability for CQ-channels;
\[ \eps \Big(\Big\lfloor \frac{\log N}{\log |\mathcal{X}|}\Big\rfloor,W\Big)
\ge 1-\eps_{\text{ID}}(N,W).\]

\subsection{Strong converse for Channel Resolvability}\Label{S7}
Next, we proceed to the strong converse for
channel resolvability for CQ channels.
To prove this statement, we show the following relation between
the identification capacity and the capacity for channel resolvability for CQ channels.

\begin{theorem}\Label{Th2}
For $\epsilon \ge 1-\lambda_1-\lambda_2$, we have
\begin{align}
C_{\lambda_1,\lambda_2}(W)  \le R_{\epsilon}(W) .
\end{align}
Thus, for $0<\epsilon<1$,
\begin{align}
C(W)  \le R_{\epsilon}(W) .
\end{align}
\end{theorem}

The above theorem implies the strong converse for
channel resolvability for CQ channels as follows.
Summarizing Theorems \ref{Th1} and \ref{Th2} with \eqref{PSE}, we obtain our main claim in Theorem~\ref{theorem:worst-input}, i.e.,
\begin{align}
R(W) =R_{\epsilon}(W) = C(W)
\end{align}
for $0<\epsilon<1 $.

\begin{proof}Recall the notation that
$$M_\epsilon(W):= \min\{M \in \mathbb{N}_+ \ |\ \eps(W,M) \le \epsilon\}.$$
Lemma \ref{lemma:bridge} implies that
\begin{align}
N(n,\lambda_1,\lambda_2) \le
(|\mathcal{X}|^n)^{M_\epsilon(W^{\otimes n})} .
\end{align}
Taking the logarithm, we have
\begin{align}
\log N(n,\lambda_1,\lambda_2) \le
M_\epsilon(W^{\otimes n}) \log (|\mathcal{X}|^n)
=n M_\epsilon(W^{\otimes n}) \log |\mathcal{X}|.
\end{align}
Taking the logarithm, again, we have
\begin{align}
\log \log N(n,\lambda_1,\lambda_2) \le
\log n+\log  M_\epsilon(W^{\otimes n})+\log \log |\mathcal{X}|.
\end{align}
Thus,
\begin{align}
\frac{1}{n}\log \log N(n,\lambda_1,\lambda_2) \le
\frac{1}{n}\log n+\frac{1}{n}\log  M_\epsilon(W^{\otimes n})+
\frac{1}{n}\log \log |\mathcal{X}|.
\end{align}
Taking the limit, we have
\begin{align*}
C_{\lambda_1,\lambda_2}(W)
=&\liminf_{n\to \infty} \frac{1}{n}\log \log N(n,\lambda_1,\lambda_2) \le
\liminf_{n\to \infty} \frac{1}{n}\log  M_\epsilon(W^{\otimes n})
=R_{\epsilon}(W).\qedhere
\end{align*}
\end{proof}

\section{Discussion}\Label{S8}
We have studied channel resolvability for CQ channels under the two problem settings,
the worst-input formulation and the fixed-input formulation.
To address the worst-input formulation, similar to existing studies \cite{han1993approximation},
we focus on its relation to identification codes for CQ channels.
While the simultaneous identification codes is a special case of identification codes,
we have observed that the direct part of identification codes with both simultaneous and non-simultaneous settings
works as the strong converse of
the worst-input channel resolvability for CQ channels.
The direct part of worst-input is shown by employing the information spectrum method for CQ channels formulated in \cite{1207373}.
To address the fixed-input formulation,
we extend the approach from \cite{watanabe2014strong} to the case with CQ channels,
in which we have used the latest result for a new version of quantum Sanov theorem, which
focuses on a quantum version of empirical distributions.

Here, we emphasize the difference between two settings.
The fixed-input formulation with the input $p$
has the minimum rate given by $\inf_{q:W(q)=W(p)}I(X;B)_{ W\times q}$,
whereas the worst-input formulation
has the minimum rate $\sup_{p}I(X;B)_{W\times p}$.
That is, the maximum value of the minimum rate of the fixed-input formulation
does not equal the minimum rate of the worst-input formulation, in general (see Example~\ref{example:separation}).
This fact means that the worst-input in the $n$-shot setting could be achieved by a non-i.i.d.~input, which shows the difficulty of the worst-input formulation. The ideas and techniques used in the work may shed a light on other converse for covering type problems such as quantum decoupling.

Recently, the paper \cite{APW} studies channel resolvability for quantum-quantum channels.
This topic is more difficult and challenging so that
the results \cite{APW}, when restricted to CQ channels, are limited to partial results. Our results of channel resolvability for CQ channels cannot be included in that for fully quantum channels
as explained in Remark \ref{REM}.
Therefore, further development of covering lemma and channel resolvability for fully quantum channels is an interesting open problem.

\section*{Acknowledgement}HC and LG are very grateful to Yu-Chen Shen for helpful discussions throughout the project.
%MH is supported in part by the National Natural Science Foundation of China (Grant No. 62171212). LG is supported in part by the National Natural Science Foundation of China.

%\input{preliminary}

%\bibliography{resolvability}

\begin{thebibliography}{10}
\providecommand{\url}[1]{#1}
\csname url@samestyle\endcsname
\providecommand{\newblock}{\relax}
\providecommand{\bibinfo}[2]{#2}
\providecommand{\BIBentrySTDinterwordspacing}{\spaceskip=0pt\relax}
\providecommand{\BIBentryALTinterwordstretchfactor}{4}
\providecommand{\BIBentryALTinterwordspacing}{\spaceskip=\fontdimen2\font plus
\BIBentryALTinterwordstretchfactor\fontdimen3\font minus
  \fontdimen4\font\relax}
\providecommand{\BIBforeignlanguage}[2]{{%
\expandafter\ifx\csname l@#1\endcsname\relax
\typeout{** WARNING: IEEEtran.bst: No hyphenation pattern has been}%
\typeout{** loaded for the language `#1'. Using the pattern for}%
\typeout{** the default language instead.}%
\else
\language=\csname l@#1\endcsname
\fi
#2}}
\providecommand{\BIBdecl}{\relax}
\BIBdecl

\bibitem{Wyn75b}
A.~Wyner, ``The common information of two dependent random variables,''
  \emph{{IEEE} Transactions on Information Theory}, vol.~21, no.~2, pp.
  163--179, mar 1975.

\bibitem{Cuff08}
P.~Cuff, ``Communication requirements for generating correlated random
  variables,'' in \emph{2008 IEEE International Symposium on Information
  Theory}.\hskip 1em plus 0.5em minus 0.4em\relax IEEE, 2008.

\bibitem{YT18}
L.~Yu and V.~Y.~F. Tan, ``Wyner’s common information under {R{\'e}nyi}
  divergence measures,'' \emph{IEEE Transactions on Information Theory},
  vol.~64, no.~5, pp. 3616--3632, May 2018.

\bibitem{YT20}
------, ``Corrections to ``{Wyner’s} common information under {R{\'e}nyi}
  divergence measures'','' \emph{IEEE Transactions on Information Theory},
  vol.~66, no.~4, pp. 2599--2608, 2020.

\bibitem{YT22_book}
------, ``Common information, noise stability, and their extensions,''
  \emph{Foundations and Trends® in Communications and Information Theory},
  vol.~19, no.~2, pp. 107--389, 2022.

\bibitem{han1993approximation}
T.~S. Han and S.~Verd{\'u}, ``Approximation theory of output statistics,''
  \emph{IEEE Transactions on Information Theory}, vol.~39, no.~3, pp. 752--772,
  1993.

\bibitem{Han03}
T.~S. Han, \emph{Information-Spectrum Methods in Information Theory}.\hskip 1em
  plus 0.5em minus 0.4em\relax Springer Berlin Heidelberg, 2003.

\bibitem{Wyn75}
A.~D. Wyner, ``The wire-tap channel,'' \emph{Bell System Technical Journal},
  vol.~54, no.~8, pp. 1355--1387, October 1975.

\bibitem{CK78}
I.~Csiszar and J.~Korner, ``Broadcast channels with confidential messages,''
  \emph{IEEE Transactions on Information Theory}, vol.~24, no.~3, pp. 339--348,
  May 1978.

\bibitem{hayashi2006general2}
M.~Hayashi, ``General nonasymptotic and asymptotic formulas in channel
  resolvability and identification capacity and their application to the
  wiretap channel,'' \emph{IEEE Transactions on Information Theory}, vol.~52,
  no.~4, pp. 1562--1575, 2006.

\bibitem{hayashi2011}
------, ``Exponential decreasing rate of leaked information in universal random
  privacy amplification,'' \emph{IEEE Transactions on Information Theory},
  vol.~57, no.~6, pp. 3989--4001, June 2011.

\bibitem{BL13}
M.~R. Bloch and J.~N. Laneman, ``Strong secrecy from channel resolvability,''
  \emph{IEEE Transactions on Information Theory}, vol.~59, no.~12, pp.
  8077--8098, 2013.

\bibitem{PTM17}
M.~B. Parizi, E.~Telatar, and N.~Merhav, ``Exact random coding secrecy
  exponents for the wiretap channel,'' \emph{{IEEE} Transactions on Information
  Theory}, vol.~63, no.~1, pp. 509--531, jan 2017.

\bibitem{AD89}
R.~Ahlswede and G.~Dueck, ``Identification via channels,'' \emph{{IEEE}
  Transactions on Information Theory}, vol.~35, no.~1, pp. 15--29, 1989.

\bibitem{Cuf13}
P.~Cuff, ``Distributed channel synthesis,'' \emph{{IEEE} Transactions on
  Information Theory}, vol.~59, no.~11, pp. 7071--7096, nov 2013.

\bibitem{Cuff}
------, ``Soft covering with high probability,'' in \emph{2016 IEEE
  International Symposium on Information Theory (ISIT)}, 2016, pp. 2963--2967.

\bibitem{LCV16}
J.~Liu, P.~Cuff, and S.~Verdu, ``{$E_\gamma$}-resolvability,'' \emph{{IEEE}
  Transactions on Information Theory}, no.~5, pp. 1--1, Mar 2017.

\bibitem{YaCu}
S.~Yagli and P.~Cuff, ``Exact exponent for soft covering,'' \emph{IEEE
  Transactions on Information Theory}, vol.~65, no.~10, pp. 6234--6262, 2019.

\bibitem{Yassaee}
M.~H. Yassaee, ``Almost exact analysis of soft covering lemma via large
  deviation,'' in \emph{2019 IEEE International Symposium on Information Theory
  (ISIT)}, 2019, pp. 1387--1391.

\bibitem{SV96}
Y.~Steinberg and S.~Verd{\'u}, ``Simulation of random processes and
  rate-distortion theory,'' \emph{IEEE Transactions on Information Theory},
  vol.~42, no.~1, pp. 63--86, 1996.

\bibitem{YAM+14}
M.~H. Yassaee, M.~R. Aref, and A.~Gohari, ``Achievability proof via output
  statistics of random binning,'' \emph{IEEE Transactions on Information
  Theory}, vol.~60, no.~11, p. 6760–6786, November 2014.

\bibitem{MBG+19}
M.~M. Mojahedian, S.~Beigi, A.~Gohari, M.~H. Yassaee, and M.~R. Aref, ``A
  correlation measure based on vector-valued $l_p$ -norms,'' \emph{IEEE
  Transactions on Information Theory}, vol.~65, no.~12, pp. 7985--8004,
  December 2019.

\bibitem{Yu24}
\BIBentryALTinterwordspacing
L.~Yu, ``R{\'e}nyi resolvability, noise stability, and anti-contractivity,''
  2024. [Online]. Available: \url{https://arxiv.org/abs/2402.07660}
\BIBentrySTDinterwordspacing

\bibitem{BSS+02}
C.~Bennett, P.~Shor, J.~Smolin, and A.~Thapliyal, ``Entanglement-assisted
  capacity of a quantum channel and the reverse {Shannon} theorem,''
  \emph{{IEEE} Transactions on Information Theory}, vol.~48, no.~10, pp.
  2637--2655, oct 2002.

\bibitem{BCR11}
M.~Berta, M.~Christandl, and R.~Renner, ``The quantum reverse {Shannon} theorem
  based on one-shot information theory,'' \emph{Communications in Mathematical
  Physics}, vol. 306, no.~3, pp. 579--615, August 2011.

\bibitem{BDH+14}
C.~H. {Bennett}, I.~{Devetak}, A.~W. {Harrow}, P.~W. {Shor}, and A.~{Winter},
  ``The quantum reverse shannon theorem and resource tradeoffs for simulating
  quantum channels,'' \emph{IEEE Transactions on Information Theory}, vol.~60,
  no.~5, pp. 2926--2959, May 2014.

\bibitem{OCC+24}
\BIBentryALTinterwordspacing
A.~Oufkir, M.~X. Cao, H.-C. Cheng, and M.~Berta, ``Exponents for shared
  randomness-assisted channel simulation,'' 2024. [Online]. Available:
  \url{https://arxiv.org/abs/2410.07051}
\BIBentrySTDinterwordspacing

\bibitem{1377491}
I.~Devetak, ``The private classical capacity and quantum capacity of a quantum
  channel,'' \emph{IEEE Transactions on Information Theory}, vol.~51, no.~1,
  pp. 44--55, 2005.

\bibitem{CWY}
N.~Cai, A.~Winter, and R.~W. Yeung, ``Quantum privacy and quantum wiretap
  channels,'' \emph{Problems of Information Transmission}, vol.~40, no.~4, pp.
  318--336, 2004.

\bibitem{hayashi2012quantum}
M.~Hayashi, ``Quantum wiretap channel with non-uniform random number and its
  exponent and equivocation rate of leaked information,'' \emph{IEEE
  Transactions on Information Theory}, vol.~61, no.~10, pp. 5595--5622, 2015.

\bibitem{loeber}
\BIBentryALTinterwordspacing
P.~L\"ober, ``Quantum channels and simultaneous {ID} coding,'' 1999. [Online].
  Available: \url{https://arxiv.org/abs/quant-ph/9907019}
\BIBentrySTDinterwordspacing

\bibitem{ahlswede2002strong}
R.~Ahlswede and A.~Winter, ``Strong converse for identification via quantum
  channels,'' \emph{IEEE Transactions on Information Theory}, vol.~48, no.~3,
  pp. 569--579, 2002.

\bibitem{Wilde17}
M.~M. Wilde, \emph{Quantum Information Theory}.\hskip 1em plus 0.5em minus
  0.4em\relax Cambridge University Press, 2017.

\bibitem{DW05}
I.~Devetak and A.~Winter, ``Distillation of secret key and entanglement from
  quantum states,'' \emph{Proceedings of the Royal Sciety A}, vol. 461, no.
  2053, pp. 207--235, 2005.

\bibitem{Win04}
A.~Winter, ``Extrinsic and intrinsic data in quantum measurements: Asymptotic
  convex decomposition of positive operator valued measures,''
  \emph{Communications in Mathematical Physics}, vol. 244, no.~1, pp. 157--185,
  January 2004.

\bibitem{Win05}
------, ``Secret, public and quantum correlation cost of triples of random
  variables,'' in \emph{2005 IEEE International Symposium on Information Theory
  (ISIT)}.\hskip 1em plus 0.5em minus 0.4em\relax IEEE, 2005.

\bibitem{cheng2023error}
H.-C. Cheng and L.~Gao, ``Error exponent and strong converse for quantum soft
  covering,'' \emph{IEEE Transactions on Information Theory}, vol.~70, no.~5,
  pp. 3499--3511, May 2024.

\bibitem{SGC22b}
Y.-C. Shen, L.~Gao, and H.-C. Cheng, ``Optimal second-order rates for quantum
  soft covering and privacy amplification,'' \emph{IEEE Transactions on
  Information Theory}, vol.~70, no.~7, p. 5077–5091, 2024.

\bibitem{SGC23}
------, ``Privacy amplification against quantum side information via regular
  random binning,'' in \emph{2023 59th Annual Allerton Conference on
  Communication, Control, and Computing (Allerton)}.\hskip 1em plus 0.5em minus
  0.4em\relax IEEE, September 2023.

\bibitem{hbook}
M.~Hayashi, \emph{Quantum Information Theory: Mathematical Foundation}.\hskip
  1em plus 0.5em minus 0.4em\relax United States: Springer Cham, 2017.

\bibitem{CHW16}
E.~Chitambar, M.-H. Hsieh, and A.~Winter, ``The private and public correlation
  cost of three random variables with collaboration,'' \emph{IEEE Transactions
  on Information Theory}, vol.~62, no.~4, pp. 2034--2043, 2016.

\bibitem{KKG+19}
\BIBentryALTinterwordspacing
S.~Khatri, E.~Kaur, S.~Guha, and M.~M. Wilde, ``Second-order coding rates for
  key distillation in quantum key distribution.'' [Online]. Available:
  \url{https://arxiv.org/abs/1910.03883}
\BIBentrySTDinterwordspacing

\bibitem{WNI}
A.~Winter, A.~C.~A. Nascimento, and H.~Imai, ``Commitment capacity of discrete
  memoryless channels,'' in \emph{Cryptography and Coding}, K.~G. Paterson,
  Ed.\hskip 1em plus 0.5em minus 0.4em\relax Berlin, Heidelberg: Springer
  Berlin Heidelberg, 2003, pp. 35--51.

\bibitem{Wil17b}
M.~M. Wilde, ``Position-based coding and convex splitting for private
  communication over quantum channels,'' \emph{Quantum Information Processing},
  vol.~16, no.~10, sep 2017.

\bibitem{LD09}
Z.~Luo and I.~Devetak, ``Channel simulation with quantum side information,''
  \emph{{IEEE} Transactions on Information Theory}, vol.~55, no.~3, pp.
  1331--1342, mar 2009.

\bibitem{ADJ17}
A.~Anshu, V.~K. Devabathini, and R.~Jain, ``Quantum communication using
  coherent rejection sampling,'' \emph{Phys. Rev. Lett.}, vol. 119, p. 120506,
  Sep 2017.

\bibitem{AARJ65}
A.~Anshu, R.~Jain, and N.~A. Warsi, ``Convex-split and hypothesis testing
  approach to one-shot quantum measurement compression and randomness
  extraction,'' \emph{IEEE Transactions on Information Theory, Volume 65 ,
  Issue 9}, 2019.

\bibitem{CG23}
\BIBentryALTinterwordspacing
H.-C. Cheng and L.~Gao, ``Tight one-shot analysis for convex splitting with
  applications in quantum information theory,'' 2023. [Online]. Available:
  \url{https://arxiv.org/abs/2304.12055}
\BIBentrySTDinterwordspacing

\bibitem{CGB23}
\BIBentryALTinterwordspacing
H.-C. Cheng, L.~Gao, and M.~Berta, ``Quantum broadcast channel simulation via
  multipartite convex splitting,'' 2023. [Online]. Available:
  \url{https://arxiv.org/abs/2304.12056}
\BIBentrySTDinterwordspacing

\bibitem{GHC23}
\BIBentryALTinterwordspacing
I.~George, M.-H. Hsieh, and E.~Chitambar, ``One-shot distributed source
  simulation: As quantum as it can get,'' 2023. [Online]. Available:
  \url{https://arxiv.org/abs/2301.04301}
\BIBentrySTDinterwordspacing

\bibitem{GC24}
I.~George and H.-C. Cheng, ``Coherent distributed source simulation as
  multipartite quantum state splitting,'' in \emph{2024 IEEE International
  Symposium on Information Theory (ISIT)}.\hskip 1em plus 0.5em minus
  0.4em\relax IEEE, 2024, pp. 1221--1226.

\bibitem{APW}
T.~A. Atif, S.~S. Pradhan, and A.~Winter, ``Quantum soft-covering lemma with
  applications to rate-distortion coding, resolvability and identification via
  quantum channels,'' \emph{International Journal of Quantum Information},
  2024.

\bibitem{watanabe2014strong}
S.~Watanabe and M.~Hayashi, ``Strong converse and second-order asymptotics of
  channel resolvability,'' in \emph{2014 IEEE International Symposium on
  Information Theory}.\hskip 1em plus 0.5em minus 0.4em\relax IEEE, 2014, pp.
  1882--1886.

\bibitem{Hol98}
A.~Holevo, ``The capacity of the quantum channel with general signal states,''
  \emph{{IEEE} Transaction on Information Theory}, vol.~44, no.~1, pp.
  269--273, 1998.

\bibitem{ON99}
T.~Ogawa and H.~Nagaoka, ``Strong converse to the quantum channel coding
  theorem,'' \emph{{IEEE} Transaction on Information Theory}, vol.~45, no.~7,
  pp. 2486--2489, 1999.

\bibitem{Win99b}
A.~Winter, ``Coding theorem and strong converse for quantum channels,''
  \emph{{IEEE} Transaction on Information Theory}, vol.~45, no.~7, pp.
  2481--2485, 1999.

\bibitem{1207373}
M.~Hayashi and H.~Nagaoka, ``General formulas for capacity of classical-quantum
  channels,'' \emph{IEEE Transactions on Information Theory}, vol.~49, no.~7,
  pp. 1753--1768, 2003.

\bibitem{bjelakovic2005quantum}
I.~Bjelakovi{\'c}, J.-D. Deuschel, T.~Kr{\"u}ger, R.~Seiler,
  R.~Siegmund-Schultze, and A.~Szko{\l}a, ``A quantum version of sanov's
  theorem,'' \emph{Communications in mathematical physics}, vol. 260, pp.
  659--671, 2005.

\bibitem{notzel2014hypothesis}
J.~N{\"o}tzel, ``Hypothesis testing on invariant subspaces of the symmetric
  group: part i. quantum sanov's theorem and arbitrarily varying sources,''
  \emph{Journal of Physics A: Mathematical and Theoretical}, vol.~47, no.~23,
  p. 235303, 2014.

\bibitem{another}
M.~Hayashi, ``Another quantum version of sanov theorem,'' \emph{arXiv preprint
  arXiv:2407.18566}, 2024.

\bibitem{CG24}
\BIBentryALTinterwordspacing
H.-C. Cheng and L.~Gao, ``On strong converse theorems for quantum hypothesis
  testing and channel coding,'' 2024. [Online]. Available:
  \url{https://arxiv.org/abs/2403.13584}
\BIBentrySTDinterwordspacing

\bibitem{Ume54}
H.~Umegaki, ``Conditional expectation in an operator algebra, {I},''
  \emph{Tohoku Mathematical Journal}, vol.~6, no. 2-3, pp. 177--181, jan 1954.

\bibitem{verdu1994general}
S.~Verd{\'u} \emph{et~al.}, ``A general formula for channel capacity,''
  \emph{IEEE Transactions on Information Theory}, vol.~40, no.~4, pp.
  1147--1157, 1994.

\bibitem{MDS+13}
M.~M{\"u}ller-Lennert, F.~Dupuis, O.~Szehr, S.~Fehr, and M.~Tomamichel, ``On
  quantum {R{\'e}nyi} entropies: A new generalization and some properties,''
  \emph{Journal of Mathematical Physics}, vol.~54, no.~12, p. 122203, 2013.

\bibitem{WWY14}
M.~M. Wilde, A.~Winter, and D.~Yang, ``Strong converse for the classical
  capacity of entanglement-breaking and {Hadamard} channels via a sandwiched
  {R{\'{e}}nyi} relative entropy,'' \emph{Communications in Mathematical
  Physics}, vol. 331, no.~2, pp. 593--622, Jul 2014.

\bibitem{HT14}
M.~Hayashi and M.~Tomamichel, ``Correlation detection and an operational
  interpretation of the {R{\'{e}}nyi} mutual information,'' \emph{Journal of
  Mathematical Physics}, vol.~57, no.~10, p. 102201, Oct 2016.

\bibitem{CGH18}
H.-C. Cheng, L.~Gao, and M.-H. Hsieh, ``Properties of noncommutative
  r{\'{e}}nyi and {Augustin} information,'' \emph{Communications in
  Mathematical Physics}, feb 2022.

\bibitem{6872563}
M.~Hayashi, ``Large deviation analysis for quantum security via smoothing of
  renyi entropy of order 2,'' \emph{IEEE Transactions on Information Theory},
  vol.~60, no.~10, pp. 6702--6732, 2014.

\bibitem{4069150}
H.~Nagaoka and M.~Hayashi, ``An information-spectrum approach to classical and
  quantum hypothesis testing for simple hypotheses,'' \emph{IEEE Transactions
  on Information Theory}, vol.~53, no.~2, pp. 534--549, 2007.

\bibitem{Notzel}
\BIBentryALTinterwordspacing
J.~Notzel, ``Hypothesis testing on invariant subspaces of the symmetric group:
  part i. quantum sanov's theorem and arbitrarily varying sources,''
  \emph{Journal of Physics A: Mathematical and Theoretical}, vol.~47, no.~23,
  p. 235303, may 2014. [Online]. Available:
  \url{https://dx.doi.org/10.1088/1751-8113/47/23/235303}
\BIBentrySTDinterwordspacing

\bibitem{Bjelakovic}
I.~Bjelakovi\'{c}, J.-D. Deuschel, T.~Kr\"{u}ger, R.~Seiler,
  R.~Siegmund-Schultze, and A.~Szko{\l}a, ``A quantum version of sanov's
  theorem,'' \emph{Commun. Math. Phys.}, vol. 260, pp. 659 -- 671, 2005.

\bibitem{Bucklew}
J.~A. Bucklew, \emph{Large Deviation Techniques in Decision, Simulation, and
  Estimation}.\hskip 1em plus 0.5em minus 0.4em\relax John Wiley \& Sons, 1990.

\bibitem{DZ}
A.~Dembo and O.~Zeitouni, \emph{Large Deviations Techniques and
  Applications}.\hskip 1em plus 0.5em minus 0.4em\relax Springer, 2010.

\bibitem{H-01}
\BIBentryALTinterwordspacing
M.~Hayashi, ``Asymptotics of quantum relative entropy from a representation
  theoretical viewpoint,'' \emph{Journal of Physics A: Mathematical and
  General}, vol.~34, no.~16, p. 3413, apr 2001. [Online]. Available:
  \url{https://dx.doi.org/10.1088/0305-4470/34/16/309}
\BIBentrySTDinterwordspacing

\bibitem{H-02}
\BIBentryALTinterwordspacing
------, ``Optimal sequence of quantum measurements in the sense of stein's
  lemma in quantum hypothesis testing,'' \emph{Journal of Physics A:
  Mathematical and General}, vol.~35, no.~50, p. 10759, dec 2002. [Online].
  Available: \url{https://dx.doi.org/10.1088/0305-4470/35/50/307}
\BIBentrySTDinterwordspacing

\bibitem{KW01}
\BIBentryALTinterwordspacing
M.~Keyl and R.~F. Werner, ``Estimating the spectrum of a density operator,''
  \emph{Phys. Rev. A}, vol.~64, p. 052311, Oct 2001. [Online]. Available:
  \url{https://link.aps.org/doi/10.1103/PhysRevA.64.052311}
\BIBentrySTDinterwordspacing

\bibitem{CM}
M.~Christandl and G.~Mitchison, ``The spectra of quantum states and the
  kronecker coefficients of the symmetric group,'' \emph{Commun. Math. Phys.},
  vol. 261, pp. 789 -- 797, 2006.

\bibitem{CHM}
M.~C. amd A.W.~Harrow and G.~Mitchison, ``Nonzero kronecker coefficients and
  what they tell us about spectra,'' \emph{Commun. Math. Phys.}, vol. 270, pp.
  575 -- 585, 2007.

\bibitem{OW}
R.~O'Donnell and J.~Wright, ``Quantum spectrum testing,'' \emph{Commun. Math.
  Phys.}, vol. 387, pp. 1 -- 75, 2021.

\bibitem{HM02a}
\BIBentryALTinterwordspacing
M.~Hayashi and K.~Matsumoto, ``Quantum universal variable-length source
  coding,'' \emph{Phys. Rev. A}, vol.~66, p. 022311, Aug 2002. [Online].
  Available: \url{https://link.aps.org/doi/10.1103/PhysRevA.66.022311}
\BIBentrySTDinterwordspacing

\bibitem{HM02b}
\BIBentryALTinterwordspacing
M.~Hayashi, ``Exponents of quantum fixed-length pure-state source coding,''
  \emph{Phys. Rev. A}, vol.~66, p. 032321, Sep 2002. [Online]. Available:
  \url{https://link.aps.org/doi/10.1103/PhysRevA.66.032321}
\BIBentrySTDinterwordspacing

\bibitem{H-24}
M.~Hayashi and Y.~Ito, ``Entanglement measures for detectability,'' 2024.

\bibitem{H-q-text}
M.~Hayashi, \emph{Group Representation for Quantum Theory}.\hskip 1em plus
  0.5em minus 0.4em\relax United States: Springer Cham, 2017.

\bibitem{AISW}
J.~Acharya, I.~Issa, N.~V. Shende, and A.~B. Wagner, ``Estimating quantum
  entropy,'' \emph{IEEE Journal on Selected Areas in Information Theory},
  vol.~1, no.~2, pp. 454--468, 2020.

\bibitem{Group1}
M.~Hayashi, \emph{A Group Theoretic Approach to Quantum Information,}.\hskip
  1em plus 0.5em minus 0.4em\relax United States: Springer Cham, 2017.

\bibitem{GW}
R.~Goodman and N.~R. Wallach, \emph{Representations and Invariants of the
  Classical Groups}.\hskip 1em plus 0.5em minus 0.4em\relax Cambridge, UK:
  Cambridge University Press, 1999.

\bibitem{nielsen2001separable}
M.~A. Nielsen and J.~Kempe, ``Separable states are more disordered globally
  than locally,'' \emph{Physical Review Letters}, vol.~86, no.~22, p. 5184,
  2001.

\end{thebibliography}
%\bibliographystyle{IEEEtran}

\end{document}